\documentclass[final]{tlp}
\pdfoutput=1
\usepackage{amsmath,amssymb}
\usepackage{amsfonts}
\usepackage{stmaryrd}
\usepackage{hyperref}
\usepackage{tikz}
\usetikzlibrary{shapes,arrows,fit,decorations.pathreplacing}

\tikzstyle{shgraph} = [x=1.5cm,->,every label/.style=degree,every 
edge/.style=edgename]
\tikzstyle{group} = [draw, rectangle, align=center, rounded corners, minimum  
size=2em]
\tikzstyle{degree} = [inner sep=0pt, label distance=0pt, font=\footnotesize]
\tikzstyle{edgename} = [draw, auto, inner sep=2pt, font=\footnotesize]
\newcommand{\shnode}[3]{\node[group,label=30:{\tiny #1},label=330:{\tiny 
#2},label=north west:{#3}]}

\let\oldproof=\proof
\renewenvironment{proof}{\intheoremtrue\oldproof}{\hspace*{1em}\hbox{\proofbox}\endtrivlist}

\newtheorem{theorem}{Theorem}[section]
\newtheorem{lemma}[theorem]{Lemma}
\newtheorem{example}[theorem]{Example}
\newtheorem{definition}[theorem]{Definition}

\newcommand{\fun}{\rightarrow}
\newcommand{\card}[1]{|#1|}
\newcommand{\Nat}{\mathbb{N}}
\newcommand{\fwp}{\wp_f}
\newcommand{\supp}[1]{\llfloor #1 \rrfloor}
\newcommand{\mwp}{\wp_m}
\newcommand{\emptymulti}{\multil \multir}
\newcommand{\multil}{\{\!\!\{}
\newcommand{\multir}{\}\!\!\}}
\newcommand{\multisum}{\uplus}
\newcommand{\bigmultisum}{\biguplus}

\newcommand{\src}{\mathrm{src}}
\newcommand{\tgt}{\mathrm{tgt}}
\newcommand{\eq}{\mathrm{Eq}}

\newcommand{\var}{\mathcal{V}}
\newcommand{\Isubst}{\mathit{ISubst}}
\newcommand{\Ren}{\mathit{Ren}}
\newcommand{\vars}{\mathit{vars}}
\newcommand{\mgu}{\mathrm{mgu}}
\newcommand{\occ}{\mathit{occ}}
\newcommand{\rng}{\mathrm{rng}}
\newcommand{\dom}{\mathrm{dom}}
\newcommand{\mgup}{\mathrm{mgu_p}}

\newcommand{\Sharing}{\mathtt{Sharing}}

\newcommand{\Lin}{\mathtt{Lin}}

\newcommand{\Def}{\mathtt{Def}}

\newcommand{\ShLinp}{\mathtt{ShLin}^2}
\newcommand{\ShLin}{\mathtt{ShLin}}

\newcommand{\Linp}{\mathtt{ShLin}^{\omega}}
\newcommand{\lp}{\mathrm{\omega}}
\newcommand{\res}{\mathit{res}}
\newcommand{\flatten}[1]{\underline{#1}}

\newcommand{\wrt}{w.r.t.~} 
\newcommand{\ie}{i.e., } 
\newcommand{\eg}{e.g.\ } 
\newcommand{\ra}{\rightarrow} 
 
\newcommand{\mG}{{\mathcal G}}
\newcommand{\edgesize}{\footnotesize}

\title{Optimal multi-binding unification for sharing and linearity analysis}

\author[G. Amato and F. Scozzari]
{GIANLUCA AMATO  and FRANCESCA SCOZZARI\\
Dipartimento di Economia, Universit\`a di Chieti-Pescara\\
\email{gamato@unich.it, fscozzari@unich.it}}

\submitted{28 September 2012}
 \revised{15 March 2013}
 \accepted{31 May 2013} 

\begin{document}

\maketitle

\label{firstpage}

\begin{abstract}
In the analysis of logic programs, abstract domains for detecting sharing 
properties are widely used. Recently the new domain $\Linp$
has been introduced to generalize both sharing and linearity information. This 
domain is endowed with an optimal abstract operator for
single-binding  unification. The authors claim that the repeated application 
of this operator is also  optimal for multi-binding unification. This is the 
proof of  such a claim.
\end{abstract}


\section{Introduction}
\label{sec:introduction}

In the abstract interpretation-based static analysis of logic programs, many 
abstract domains for encoding sharing information have been proposed. For 
instance, in the original and most studied  $\Sharing$ domain of
\citeN{JacobsL92}, the substitution $\theta=\{x/s(u,v), y/g(u,u,u), z/v\}$ is 
abstracted into $\{uxy,vxz\}$, where the sharing group $uxy$ 
means that $\theta(u)$, $\theta(x)$, and $\theta(y)$ share a common variable, 
namely $u$. 

Since $\Sharing$ is not very precise, it is often combined with other domains 
handling freeness, linearity, groundness or structural information (see 
\citeNP{BagnaraZH05TPLP} for a comparative evaluation). In particular, adding 
some kind of linearity information seems to be very profitable, both for the 
gain in precision and speed which can be obtained, and for the fact that it 
can be easily and elegantly embedded inside the sharing groups (see 
\citeNP{King94}). For example, if we know  that $x$, $y$ and $z$ do not share, 
nothing can be said after the unification with $\{z/t(x,y)\}$. However, if we 
also know that $z$ is linear, then we may conclude that $x$ and $y$ do not 
share after the unification.

Recently, the new abstract domain $\Linp$  has been proposed by
\citeN{AmatoS09sharlin} as a  generalization of $\Sharing$. 
It is able to encode the \emph{amount} of non-linearity in a 
substitution, by keeping track of 
the exact number of occurrences of the 
same variable in a term.  The above substitution $\theta$ is abstracted 
into $\{uxy^3,vxz\}$ by $\Linp$, with the the additional information that the 
variable $u$ occurs 
three times in $\theta(y)$. The authors provide a constructive 
characterization of the optimal abstract unification operator for 
single-binding substitutions (i.e., substitutions $\{x/t\}$ with a single 
variable $x$). Such operator is used to 
derive optimal (single-binding) abstract unification operators for both the 
domains $\Sharing \times \Lin$ \cite{HW92,MuthukumarH92} and $\ShLinp$ 
\cite{King94}. These were the first optimality results for domains   
combining aliasing and linearity information.

In the same paper the authors claim that computing abstract unification over 
$\Linp$ one binding at a time yields the best abstract unification for 
multi-binding  substitutions. In this paper we prove this claim.

To this purpose, we introduce a parallel abstract unification operator, which
computes the abstract unification over $\Linp$ by  considering all the 
bindings 
at the same time. We prove that (1) the parallel unification operator and the 
standard (sequential) one do coincide over $\Linp$ and that (2) the parallel 
unification operator is optimal. 

%We have already shown in  \cite{AmatoS09sharlin} that this result does not 
%necessarily extend to other domains for sharing analysis.

\section{Preliminaries}
\label{sec:notation}
Given a set $A$, we use $\wp(A)$ for the powerset of $A$, $\fwp(A)$ for
the set of finite subsets of $A$ and $\card{A}$ for the cardinality of $A$.
$\Nat$ is the set of natural numbers with zero.
%Given sets $A$ and $B$, we denote with $A \fun B$ the set of function from 
%$A$ to $B$ and with $A \uplus B$ their disjoint union.

\subsection{Multisets}

A \emph{multiset} is a set where repetitions are allowed. We denote by 
$\multil v_1, \ldots, v_m \multir$ a multiset, where $v_1,
\ldots, v_m$ is a sequence with (possible) repetitions. 
We denote by $\emptymulti$ the empty multiset.
We will often use the polynomial notation $v_1^{i_1} \ldots
v_n^{i_n}$, where $v_1, \ldots, v_n$ is a sequence without
repetitions, to denote a multiset $A$ whose element $v_j$ appears $i_j$ times.
The set $\{v_j \mid i_j > 0\}$  is called the \emph{support} of $A$ and is 
denoted by $\supp{A}$.  We also use the functional notation $A: \{v_1, \ldots, 
v_n\} \fun \Nat$, where  $A(v_j)=i_j$. 

%The cardinality of a multiset is $|A|= 
%\sum_{v \in \supp{A}} A(v)$. 

In this paper, we only consider multisets whose support is \emph{finite}.
We denote with $\mwp(X)$ the set of all
the multisets whose support is \emph{any finite subset} of $X$. 
For example, both $a^2c^4$ and $a^1b^2c^3$ are elements of
$\mwp(\{a,b,c\})$.

The new fundamental operation for multisets
is the \emph{sum}, defined as
\begin{equation}
  A \multisum B = \lambda v \in \supp{A} \cup \supp{B}. A(v)+B(v)
  \enspace .
\end{equation}
For instance, the sum of $a^2c^4$ and $a^1b^2c^3$ is $a^3b^2c^7$. 
%Multiset sum is associative, commutative and $\emptymulti$ is the
%neutral element.
%Note that we also use $\multisum$ to denote disjoint
%union for standard sets. The context will allow us to identify the
%proper semantics of $\uplus$.
Given a multiset $A$ and $X \subseteq \supp{A}$,
the \emph{restriction} of $A$ over $X$, denoted by
$A|_{X}$, is the only multiset $B$ such that $\supp{B}=X$ and $B(v)=A(v)$ 
for each $v \in X$. 

%Finally, if $A \in \mwp(X)$ and $E[x]$ is a numeric 
%expression containing a variable $x$, we define
%\begin{equation}
%\sum_{x \in A} E[x] = \sum_{x \in \supp{A}} A(x) \cdot E[x] \enspace .
%\end{equation}
%For example, given $A=\multil 1,1,2,2,2,3 \multir$, $\sum_{x \in A} 
%x^3 = 2 \cdot 1^3+ 3\cdot 2^3+ 3^3=53$.
%%  If $X$ is a multiset of multisets, i.e.  $X=\multil X_1,\ldots,
%% X_n \multir$, we also write $\multisum X = \multisum_{i \in I} X_i$ or,
%% with the same meaning, $X_1 \multisum \ldots \multisum X_n$. If a set $S$ is
%% used in a context where a multiset is expected, it stands for the
%% multiset with support $S$ and such that $S(x)=1$ for each $x \in S$.

\subsection{Multigraphs}
We call (directed) \emph{multigraph} a graph where multiple distinguished 
edges are allowed between nodes.  We use the definition of multigraph which is 
customary in category theory \cite{MacLane71}.

\begin{definition}[Multigraph]
  A \emph{multigraph} $G$ is a tuple $\langle N_G, E_G, \src_G, \tgt_G
  \rangle$ where $N_G \neq \emptyset$ and $E_G$ are the sets of
  \emph{nodes} and \emph{edges} respectively, $\src_G: E_G \fun N_G$
  is the \emph{source} function which maps each edge to its starting
  node, and $\tgt_G: E_G \fun N_G$ is the \emph{target} function which
  maps each edge to its ending node.
%  
%  A \emph{labeled} multigraph $G$ is a multigraph equipped with a
%  \emph{labeling} function $l_G: N_G \ra L_G$ which maps each node to
%  its \emph{label} in the given set $L_G$.
\end{definition}

%We write $e:n_1 \ra n_2 \in G$ to denote the edge $e \in E_G$ such
%that $\src_G(e)=n_1$ and $\tgt_G(e)=n_2$. 
We write $e: n_1 \ra n_2
\in G$ to denote an edge $e \in E_G$ such that $\src_G(e)=n_1$ and
$\tgt_G(e)=n_2$. 
% Moreover, with $| n_1 \ra n_2 \in G |$ we denote the
%cardinality of the set $\{ e \in E_G \mid \src_G(e)=n_1 \wedge
%\tgt_G(e)=n_2 \}$. In the notation above, we omit ``$\in G$'' whenever
%the multigraph $G$ is clear from the context.
We call \emph{in-degree} (respectively \emph{out-degree}) of a node
$n$ the cardinality of the set $\{ e \in E_G \mid \tgt_G(e)=n \}$
(respectively $\{ e \in E_G \mid \src_G(e)=n \}$).

Given a multigraph $G$, a \emph{path} $\pi:n_1 \ra n_{k}$ is a \emph{non-empty}
sequence of nodes $n_1 \ldots n_k$ such that, for each $i \in
\{1,\ldots,k-1\}$, there is either an edge $n_i \ra n_{i+1} \in G$ or
an edge $n_{i+1} \ra n_{i} \in G$. Nodes $n_1$ and $n_k$ are the
\emph{endpoints} of $\pi$, and we say that $\pi$ \emph{connects} $n_1$
and $n_k$. A multigraph is \emph{connected} when all pairs of nodes are
connected by at least one path.

\begin{example}
Consider the multigraph $G$ such that $N_G = \{1,2,3\}$, $E_G = \{a,b,c,d,e\}$,
$\src_G = \{ a \mapsto 1, b \mapsto 1, c \mapsto 2, d \mapsto 2, e \mapsto 
1\}$ and $\tgt_G = \{ a \mapsto 1, b \mapsto 2, c \mapsto 3, d \mapsto 3, e 
\mapsto 3\}$. It may be depicted as follows:
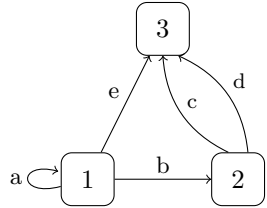
\[
\begin{tikzpicture}[->,auto,inner sep=2pt]
\node[group] (v1) at (0,0) {1};
\node[group] (v2) at (2,0) {2};
\node[group] (v3) at (1,2) {3};
\path (v1) edge[loop left] node{\small a} (v1);
\path (v1) edge node{\small b} (v2);
\path (v1) edge node{\small e} (v3);
\path (v2.110) edge[swap,bend left] node{\edgesize c} (v3.270);
\path (v2.70) edge[swap,bend right] node{\edgesize d} (v3.300);
\end{tikzpicture}
\]
Note that the edges $c$ and $d$ have the same starting and ending nodes, but different names.
According to our definition, the graph is connected. 
\end{example}

\subsection{Abstract Interpretation}

Given two sets $C$ and $A$ of concrete and abstract objects
respectively, an \emph{abstract interpretation} \cite{CousotC92fr} is
given by an approximation relation $\rightslice \subseteq A \times C$.
When $a \rightslice c$ holds, this means that $a$ is a correct
abstraction of $c$. We work in a framework where: (1) $A$ is a complete 
lattice, (2) $a \rightslice c$ and $a \leq a'$ imply $a' \rightslice c$,  (3) 
each $c$ has a least correct abstraction in $A$ given by 
$\alpha(c)$.

Given a function $f: C \ra C$, we say that $\tilde{f}: A \ra A$ is a correct 
abstraction of $f$, and we write $\tilde{f} \rightslice f$, when
\begin{equation*}
  a \rightslice c  \Rightarrow \tilde{f}(a)
  \rightslice f(c) \enspace .
\end{equation*}
We say that $\tilde{f}: A \ra A$ is the \emph{optimal} abstraction of $f$ when 
it is correct and, for each $f': A \ra A$,
\begin{equation*}
   f' \rightslice f
  \Rightarrow \tilde f \leq f'
\end{equation*}
with the standard pointwise ordering.
%\begin{equation*}
%  \forall f':A \ra A.\  f' \rightslice f
%  \Rightarrow \forall a \in A. \tilde{f}(a) \leq f'(a) \enspace .
%\end{equation*}

\subsection{Terms and Substitutions}

In the following, we fix a first order signature and a denumerable set
of variables $\var$. Given a term or other syntactic object $o$, we
denote by $\vars(o)$ the set of variables occurring in $o$ and by
$\occ(v,o)$ the number of occurrences of $v$ in $o$. When it does not
cause ambiguities, we abuse the notation and prefer to use $o$ itself
in the place of $\vars(o)$. For example, if $t$ is a term and $x \in
\var$, then $x \in t$ should be read as $x \in \vars(t)$.

We denote by $\epsilon$ the empty substitution, by $\{x_1/t_1,
\ldots, x_p/t_p\}$ a substitution $\theta$ with $\theta(x_i)=t_i\neq
x_i$, and by $\dom(\theta)=\{x\in\var \mid \theta(x)\neq x\}$ and 
$\rng(\theta)=\cup_{x\in\dom(\theta)}\vars(\theta(x))$ the domain and range of
$\theta$ respectively. Let $\vars(\theta)$ be the set $\dom(\theta)\cup
\rng(\theta)$, and given $U\in \fwp(\var)$, let $\theta|_{U}$ be the
projection of $\theta$ over $U$, \ie the unique substitution such that
$\theta|_{U}(x)=\theta(x)$ if $x\in U$ and $\theta|_{U}(x)=x$
otherwise. 
%% We also write $\theta|_{-U}$ to denote the restriction of
%% $\theta$ over all variables but those in $U$, \ie
%% $\theta|_{-U}=\theta|_{\dom(\theta) \setminus U}$.  
Given $\theta_1$ and $\theta_2$ two substitutions with disjoint
domains, we denote by $\theta_1 \cup \theta_2$ the substitution
$\theta$ such that $\dom(\theta)=\dom(\theta_1) \cup \dom(\theta_2)$
and $\theta(x)=\theta_i(x)$ if $x \in \dom(\theta_i)$, for each $i \in
\{1,2\}$.  The application of a substitution $\theta$ to a term $t$ is
written as $t \theta$ or $\theta(t)$. Given two substitutions $\theta$
and $\delta$, their composition, denoted by $\theta\circ \delta$, is
given by $(\theta \circ \delta)(x)=\delta(\theta(x))$. A substitution $\theta$ 
is idempotent when $\theta \circ \theta=\theta$ or, equivalently, when 
$\dom(\theta)\cap\rng(\theta)=\emptyset$. A substitution
$\rho$ is called renaming if it is a bijection from $\var$ to $\var$
(this is equivalent to saying that there exists a substitution
$\rho^{-1}$ such that $\rho \circ \rho^{-1} = \rho^{-1} \circ \rho =
\epsilon$).
%Instantiation induces a preorder on substitutions:
%$\theta$ is more general than $\delta$, denoted by $\delta
%\leq\theta$, if there exists $\eta$ such that $\eta \circ \theta =
%\delta$. If $\approx$ is the equivalence relation induced by $\leq$,
%we say that $\delta$ and $\theta$ are equal up to renaming when
%$\delta \approx \theta$. 
The sets of idempotent substitutions and renamings are denoted by $\Isubst$
and $\Ren$ respectively. Given a set of 
equations $E$, we write $\theta=\mgu(E)$ to denote that $\theta$ is a most 
general unifier of $E$.
%% Since $\theta$ is only defined up to renamings, we only use this
%% notation in a context where it is not important the actual unifier
%% which is chosen.
Conversely, $\eq(\theta)=\{ x = \theta(x) \mid
x\in\dom(\theta)\}$. 
%In the following, we will abuse the notation and
%denote by $\mgu(\theta_1,\ldots,\theta_n)$ the
%substitution $\mgu(\eq(\theta_1) \cup \ldots \cup \eq(\theta_n))$, when it 
%exists.

A \emph{position} is a sequence of positive natural numbers. 
Given a term $t$ and a position  $\xi$, we define $t(\xi)$ inductively as 
follows:
\begin{equation}
 \begin{split}
  t(\epsilon) &= t \qquad \text{(where $\epsilon$ 
    denotes the empty sequence)}\\
  t(i \cdot \xi')&= \begin{cases} t_i(\xi') & \text{if $t$ is
      $s(t_1, \ldots, t_p)$ and $i \leq p$;}\\
    \text{undefined} & \text{otherwise.}
  \end{cases}
  \end{split}
\end{equation}
For any variable $x$, an \emph{occurrence} of $x$ in
$t$ is a position $\xi$ such that $t(\xi)=x$.

In the rest of the paper, we use: $U$, $V$, $W$ to denote finite sets
of variables; $u,v,w,x,y,z$ for variables; $r, s$ for term symbols; $t$ for
terms; $\beta, \eta,\theta,\delta$ for substitutions; $\rho$ for 
renamings.
%
%\section{Preliminaries}
%
%This section recalls some preliminary work which appeared in 
%\cite{AmatoS09sharlin} and \cite{AmatoS09sharing}. We refer the reader to 
%these papers for proofs and further details.

\subsection{Existential Substitutions}
\label{sec:existential}

The denotational semantics of logic programs is not generally interested in  
substitutions, but in appropriate equivalence classes which abstract away from 
the particular renaming of clauses used during SLD derivations. Among the 
many choices available in the literature  (\eg 
\citeNP{JacobsL92,MSJ94,LeviS03}), we 
adopt the domain of existential substitutions \cite{AmatoS09sharing}.
% We 
%briefly recall the basic definitions of the domain and the unification 
%operator.

Given $\theta_1, \theta_2 \in \Isubst$ and $U \in \fwp(\var)$, consider the 
equivalence relation $\sim_U$ given by
\begin{equation}\label{eq:2}
  \theta_1 \sim_U \theta_2 \iff \exists \rho \in \Ren. 
  \forall v \in U.\ \theta_1(v)=\rho(\theta_2(v)) \enspace ,
\end{equation}
and let $\Isubst_{\sim_U}$ be the quotient set of $\Isubst$ \wrt
$\sim_U$. The domain $\Isubst_\sim$ of \emph{existential
  substitutions} is defined as the union of all the
$\Isubst_{\sim_U}$ for $U\in \fwp(\var)$, namely:
\begin{equation}
\Isubst_\sim = \bigcup_{U\in \fwp(\var)} \Isubst_{\sim_U} \enspace .
\end{equation}
In the following we write $[\theta]_{U}$ for the equivalence class of
$\theta$ \wrt $\sim_U$. 
To ease notation, we often omit braces from the sets of variables of
interest when they are given extensionally. So we write
$[\theta]_{x,y}$ instead of $[\theta]_{\{x,y\}}$.
%and $\sim_{x,y,z}$ instead of $\sim_{\{x,y,z\}}$.  
%When the set of variables of interest
%is clear from the context or when it is not relevant, it will be omitted.
%Finally, we omit the braces which enclose the bindings of a
%substitution when the latter occurs inside an equivalence class, \ie we write
%$[x/y]_U$ instead of $[\{x/y\}]_U$.

Given $U \in \fwp(\var)$, $[\delta]_U \in \Isubst_{\sim}$ and $\theta \in 
\Isubst$, the most general unifier of $\theta$ and $[\delta]_U$ may be 
obtained from the  mgu of $\theta$ and a suitably chosen representative for 
$\delta$, where variables not of interest are renamed apart. In formulas:
\begin{equation}
  \label{eq:mgu}
  \mgu([\delta]_U,\theta)=
  [\mgu(\delta', \theta)]_{U \cup \vars(\theta)} \enspace ,
\end{equation}
where $\delta \sim_U \delta' \in \Isubst$ and $\vars(\delta') \cap 
\vars(\theta) \subseteq U $.
%The last condition is needed 
%to avoid variables clashes between the chosen representatives $\theta'$
%and $\delta'$.

\subsection{The Domain $\Linp$}

The domain $\Linp$ \cite{AmatoS09sharlin}  generalizes $\Sharing$ by 
recording  multiplicity of 
variables in sharing groups. We will call  a 
multiset of variables (an element of $\mwp(\var)$) an \emph{$\omega$-sharing 
group}.
%\begin{example}
%Given $u,v,w,x,y\in \var$, examples of $\omega$-sharing groups are 
%$u^2v^3x^{19}$, $xyz$ and $u^{23}vwx^2y^3$.
%\end{example}
Given a substitution $\theta$ and a variable $v \in \var$, we denote
by $\theta^{-1}(v)$ the $\omega$-sharing group $\lambda w \in \var.
\occ(v,\theta(w))$, which maps each variable $w$ to the number of
occurrences of $v$ in $\theta(w)$.

  Given a set of variables $U$ and a set of $\omega$-sharing groups $S 
  \subseteq \mwp(U)$, we say that $[S]_U$ \emph{correctly
  approximates} a substitution $[\theta]_W$ if $U=W$ and for each $v
  \in \var$, $\theta^{-1}(v)|_U \in S$. We write $[S]_U \rightslice
  [\theta]_W$ to mean that $[S]_U$ correctly approximates $[\theta]_W$.
Therefore, $[S]_U \rightslice [\theta]_U$ when $S$ contains all the 
$\omega$-sharing groups in $\theta$, restricted to the 
variables in $U$. 

\begin{definition}[$\Linp$]
The domain $\Linp$ is defined as
\begin{equation}
   \Linp =\{ [S]_U \mid U \in \fwp(\var), S \subseteq \mwp(U),
  S \neq \emptyset \Rightarrow \emptymulti \in S\} \enspace ,
\end{equation}
and ordered by $[S_1]_{U_1} \leq_\omega [S_2]_{U_2}$ iff $U_1=U_2$ and
$S_1 \subseteq S_2$. 
\end{definition}
%The existence of the empty multiset, when $S$ is not empty, is required in 
%order to a surjective abstraction function.
In order to ease the notation, we write $[\{\emptymulti, B_1,\ldots, B_n\}]_U$ 
as  $[B_1,\ldots, B_n]_U$ by omitting the braces and the empty 
multiset.  Moreover, if $X \in \Linp$, we write
$B \in X$ in place of $X=[S]_U \wedge B \in S$. The best correct abstraction 
of a substitution $[\theta]_U$ is
\begin{equation}
  \alpha_\omega([\theta]_U)=[\{ \theta^{-1}(v)|_U \mid v \in \var \}]_U
  \enspace .
\end{equation}

\begin{example}
  Given $\theta=\{x/s(y,u,y), z/s(u,u), v/u\}$ and $U=\{w,x,y,z\}$, we have
  $\theta^{-1}(u)=uvxz^2$, $\theta^{-1}(y)= x^2y$,
  $\theta^{-1}(z)=\theta^{-1}(v)=\theta^{-1}(x)=\emptymulti$ and
  $\theta^{-1}(v)=v$ for all the other variables (included $w$).
  Projecting over $U$ we obtain
  $\alpha_\omega([\theta]_U)=[xz^2,x^2y,w]_U$.
\end{example}

\begin{definition}[Multiplicity of $\omega$-sharing groups]
The \emph{multiplicity} of an $\omega$-sharing
group $B$ in a term $t$ is defined as:
\begin{equation}
  \label{eq:chidef}
  \chi(B,t)=
  \sum_{v \in \supp{B}} B(v) \cdot \occ(v,t)  \enspace .
\end{equation}
\end{definition}
For instance,
$\chi(w^2x^3yz^4,r(x,y,s(x,y,z),v))= 2 \cdot 0 + 3 \cdot 2+1\cdot 2+ 4\cdot 1 
= 12$.
%The
%meaning of the map $\chi$ is made clear by the following proposition.
%
%\begin{proposition}
%  \label{prop:chi}
%  Given a substitution $\theta$, a
%  variable $v$ and a term $t$, we have that
%  $\chi(\theta^{-1}(v),t)=\occ(v,\theta(t))$. Moreover, given a set of 
%  variables $U$, when $\vars(t) \subseteq U$, it holds that 
%  $\chi(\theta^{-1}(v)|_U,t)=\occ(v,\theta(t))$.
%\end{proposition}
%
%\begin{example}
%  Let $B=xy^2z^3$ and $\theta=\{ y/r(x,x), z/r(x,x,x) \}$, 
%  so that $\theta^{-1}(x)=\{xy^2z^3\}$. Given $t \equiv s(x,z)$ we have
%  \[
%  \occ(x,\theta(t))=\occ(x,s(x,r(x,x,x)))=4 \enspace ,
%  \]
%  and 
%  \[
%  \chi(B,t)=B(x) \occ(x,t) + B(z) \occ(z,t)=1 \cdot 1 + 3 \cdot 1=4
%  \enspace .
%  \]
%\end{example}

\section{Parallel Abstract Unification}

We want to find the optimal abstract operator in $\Linp$ corresponding to 
unification. \citeN{AmatoS09sharlin} define the operator $\mgu_\lp$, which
is optimal for single-binding substitutions.
The  cornerstone of their abstract unification is 
the concept of \emph{sharing graph} which plays the same role of alternating 
paths \cite{Son:ESOP86,King00jlp} for pair sharing analysis. The authors
claim that, by applying $\mgu_\lp$
one binding at a time, we get an optimal operator for multi-binding 
substitutions.

Here, in order to prove this claim, we proceed along these steps:
\begin{enumerate}
\item we define a new operator $\mgup$ which computes the abstract 
unification
  with a multi-binding substitution in one step. This is based on a
  generalization of the concept of sharing graph with multiple
  layers. For this reason, we speak of \emph{parallel sharing graph}
  and \emph{parallel abstract unification};
\item we prove that parallel abstract unification ($\mgup$) is actually the 
same as the \emph{sequential abstract unification} ($\mgu_\lp$);
\item we prove that parallel abstract unification is optimal
  w.r.t.~concrete unification.
\end{enumerate}

If $[S]_U \rightslice [\delta]_U$ and we unify $[\delta]_U$ with
$\theta$, some of the $\omega$-sharing groups in $S$ may be glued
together to obtain a bigger resultant group.  It happens that the gluing
may be represented by special families of labeled multigraphs which we call 
\emph{parallel sharing graphs}. 

\begin{definition}[Parallel sharing graph]\label{def:parallelSharinGraph}
  A \emph{parallel sharing graph} for a set of $\omega$-sharing groups  $S$ 
  and the idempotent substitution   $\theta=\{x_1/t_1,\ldots,
  x_p/t_p\}$ is a family $\mathcal G=\{G^i\}_{i\in [1,p]}$ of
  multigraphs over the same set of nodes $N_{\mG}$, equipped with a
  labeling function $l_{\mG}: N_{\mG} \ra S$, such that
  \begin{itemize}
   \item for each node $n \in N_{\mG}$ and each $i \in [1,p]$, the
      out-degree of $n$ in $G^i$ is equal to $\chi(l_{\mG}(n),x_i)$ and
      the in-degree of $n$ in $G^i$ is equal to $\chi(l_{\mG}(n),t_i)$;
  \item the sets of edges $E_{G^i}$ are all pairwise disjoint;
  \item $\flatten{\mG}$ (the \emph{flattening} of $\mG$) is connected.
  \end{itemize}

  In the last condition, $\flatten{\mG}$ is defined as the multigraph 
  $\langle N_{\mG}, E, \src_{\mG},\tgt_{\mG} \rangle$ where
  $E=\cup_{i \in [1,p]} E_{G^i}$ and $\src_{\mG}: E \ra N_{\mG}$
  maps $x \in E_{G^i}$ to $\src_{G^i}(e)$ ($\tgt_\mG$ is defined
  analogously). Each of the $G^i$'s which make up $\mG$ is called a
  \emph{layer} of the sharing graph.
\end{definition}

Since in this paper we only use parallel sharing graphs, in the following we 
will call them just \emph{sharing graphs}.

\begin{example}
  \label{ex:psharing}
  Let $S=\{ u^2z, uyz, vx, yz \}$ and $\theta=\{x/y,u/r(z)\}$. Consider the 
  sharing graph $\mG = \{G^1, G^2\}$ over the set of nodes 
  $N_\mG=\{a,b,c,d,e\}$ labeled by $l_\mG = \{ a \mapsto vx, b \mapsto vx, c 
  \mapsto u^2z, d \mapsto yz, e \mapsto uyz \}$:
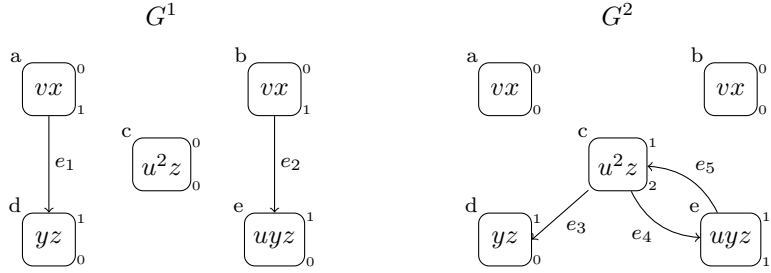
  \[
  \begin{tikzpicture}[shgraph]
   \node at (1,1) {$G^1$};
   \shnode 0 1 a (xvl) at (0,0) {$vx$};
   \shnode 1 0 d (yz)  at (0,-2) {$yz$} ;
   \shnode 0 1 b (xvr) at (2,0) {$vx$};
   \shnode 0 0 c (u2z) at (1,-1) {$u^2z$};
   \shnode 1 0 e (yuz) at (2,-2) {$uyz$};
   \path (xvl) edge node{$e_1$} (yz);
   \path (xvr) edge node{$e_2$} (yuz);
  \end{tikzpicture}
  \hspace{2cm}
  \begin{tikzpicture}[shgraph]
    \node at (1,1) {$G^2$};
    \shnode 0 0 a (xvl) at (0,0) {$vx$};
    \shnode 1 0 d (yz)  at (0,-2) {$yz$} ;
    \shnode 0 0 b (xvr) at (2,0) {$vx$};
    \shnode 1 2 c (u2z) at (1,-1) {$u^2z$};
    \shnode 1 1 e (yuz) at (2,-2) {$uyz$};
    \path (u2z) edge node{$e_3$} (yz.0);
    \path (u2z) edge[bend right] node[swap]{$e_4$} (yuz);
    \path (yuz) edge[bend right] node[swap]{$e_5$} (u2z);
  \end{tikzpicture}  
  \]  
  The left layer ($G^1$) is for the binding $x/y$, while the right one ($G^2$) 
  is for the binding $u/r(z)$. Each node is annotated with its name,  
  label, in- and out-degree.
%  When in- and out-degrees are not relevant, we will depict a  
%sharing 
%  graph with a single picture, without annotations and with arrows in 
%  different styles, according to the graph they belong to, as in:
%  \[
%    \begin{tikzpicture}[shgraph]
%     \node at (1,1) {$\mG$};
%     \mynode{}{}[label=135:\tiny a] (xvl) at (0,0) {$xv$};
%     \mynode{}{}[label=135:\tiny d] (yz)  at (0,-2) {$yz$} ;
%     \mynode{}{}[label=135:\tiny b] (xvr) at (2,0) {$xv$};
%     \mynode{}{}[label=135:\tiny c] (u2z) at (1,-1) {$u^2z$};
%     \mynode{}{}[label=135:\tiny e] (yuz) at (2,-2) {$yuz$};
%     \draw (xvl) edge[dashed] node[auto]{\tiny $e_1$} (yz);
%     \draw (xvr) edge[dashed]  node[auto]{\tiny $e_2$} (yuz);
%      \draw (u2z) edge node[auto]{\tiny $e_3$} (yz.0);
%      \draw (u2z) edge[bend right]  node[auto,swap]{\tiny $e_4$} (yuz);
%      \draw (yuz) edge[bend right]  node[auto,swap]{\tiny $e_5$} (u2z);
%    \end{tikzpicture}
%  \]
  Its flattening is the following connected multigraph:
  
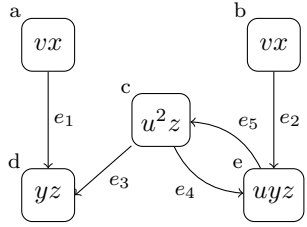
\[
    \begin{tikzpicture}[shgraph]
      \shnode{}{} a (xvl) at (0,0) {$vx$};
      \shnode{}{} d (yz)  at (0,-2) {$yz$} ;
      \shnode{}{} b (xvr) at (2,0) {$vx$};
      \shnode{}{} c (u2z) at (1,-1) {$u^2z$};
      \shnode{}{} e (yuz) at (2,-2) {$uyz$};
      \path (u2z) edge node {$e_3$} (yz.0);
      \path (xvl) edge node {$e_1$} (yz);
      \path (xvr) edge  node {$e_2$} (yuz);
      \path (u2z) edge[bend right] node[swap]{$e_4$} (yuz);
      \path (yuz) edge[bend right] node[swap]{$e_5$} (u2z);
      \end{tikzpicture}  
  \]
\end{example}

Let us motivate the three conditions of 
Definition~\ref{def:parallelSharinGraph}. If $[S]_U \rightslice [\delta]_U$ 
and we compute $\mgu([\delta]_U,\theta)$, each $G^i$ represents a possible way 
the sharing groups in $\delta$ may be joined together as a result of 
binding $x_i/t_i$, that is unifying  $\delta(x_i)$ and $\delta(t_i)$. We may 
restrict our 
attention to the case when, as a result of the unification, variables are only 
bound to other variables, not to composed terms. In other words, we assume 
that, for  each position $\xi$, the term $\delta(x_i)(\xi)$ is a variable iff 
$\delta(t_i)(\xi)$ is a variable. Each node in the sharing graph represents a 
variable $w_n$ such that $\delta^{-1}(w_n)|_U$ is the node label. Each edge 
$e: n_1 \ra n_2$ in $G^i$  represents a position $\xi$ 
such that $\delta(x_i)(\xi)=w_{n_1}$ and $\delta(t_i)(\xi)=w_{n_2}$. The 
result is that the variables $w_{n_1}$ and $w_{n_2}$ are aliased, hence the 
$\lp$-sharing groups 
%$l_{\mG}(n_1)=\delta^{-1}(w_{n_1})|_U$ and 
%$l_{\mG}(n_2)=\delta^{-1}(w_{n_2})|_U$  are joined together. 
$l_{\mG}(n_1)$ and $l_{\mG}(n_2)$  are joined together. 

According to this correspondence, the number of edges departing from $n$ 
should be equal to the number of occurrences of $w_n$  in $\delta(x_i)$,
that is $\chi(l_{\mG}(n),x_i)$. Analogously for the in-degree of nodes. This 
justifies the first condition in the definition.

The second condition  ensures that, in the flattening, no edges share the same  
identifier and therefore $\src_\mG$ and $\tgt_\mG$ are well defined. Remember 
that, since an edge is just an element of a set 
with associated source and target nodes, this does not preclude the 
possibility to have different edges with the same source and target nodes.

Finally, the third condition is needed since we want each sharing graph to 
represent a single non-empty sharing group. If the flattening were not 
connected, some pairs of variables would not be aliased, and the result of the 
unification of $\theta$ with $\delta$ would contain more than one non-empty 
sharing group.

\begin{example}
  \label{ex:psharing2}
  Consider the sharing graph in Example \ref{ex:psharing}.
  Let us associate to each node $n$ the variable 
  $w_n$, and consider the substitution 
  \[
  \delta=\{u/r(s(w_c,w_c,w_e)), 
  v/s(w_a,w_b), x/s(w_a,w_b), y/s(w_d, w_e), z/s(w_d,w_e,w_c)\} \enspace .\]
  This substitution is built according to the variables that appear in the 
  nodes. For instance, the first binding $u/r(s(w_c,w_c,w_e))$ suggests 
  that the variable $u$ appears in the nodes $c$ (twice) and $e$.
  
  We now want to unify $\delta$ with $\theta$. The first binding $x/y$ in 
  $\theta$ unifies $\delta(x) = s(w_a,w_b)$ with $\delta(y)=s(w_d, 
  w_e)$. This causes variables $(w_a, w_d)$ and $(w_b, w_e)$ to be aliased,
  exactly as described by the arrows $e_1$ and $e_2$ in the left graph.  
  The second binding
  $u/r(z)$ unifies $\delta(u)=r(s(w_c,w_c,w_e))$ with 
  $\delta(r(z))=r(s(w_d,w_e,w_c))$, which yields the aliasing of the pairs 
  $(w_c,w_d)$, $(w_c,w_e)$ and $(w_e,w_c)$, as described by the arrows $e_3$, 
  $e_4$ and $e_5$.  By transitivity, all pairs of variables are aliased. 
\end{example}

\begin{definition}[Resultant $\omega$-sharing group]
  The \emph{resultant $\omega$-sharing group} of the sharing graph 
  $\mG$ is
  \begin{equation}
    \label{eq:res}
    \res(\mG)=\bigmultisum_{s \in N_{\mG}} l_{\mG}(s) \enspace .
  \end{equation}  
\end{definition}

%$\flatten{\mG}$ is a (plain) sharing graph and 
%$\res(\mG)=\res(\flatten{\mG})$. 
%Moreover, if  $\delta=\{x/t\}$ is a single-binding substitution, then the
%flattening operation is a bijective correspondence from  sharing 
%graphs to sharing graphs.
%
%For this reason, in the following, we will speak 
%generically of sharing graphs both in the single and multi-binding cases.

%\newcommand{\mynode}[2]{\node[group,label=30:{\tiny #1},label=330:{\tiny #2}]}

\begin{example}
  Consider again the sharing graph in Example \ref{ex:psharing}.
  The resultant sharing group is $u^3v^2x^2y^2z^3$. This is
  exactly the only non-empty sharing group in $\alpha_\lp([\eta]_U)$ where 
  $U=\vars(S)$ and
  \begin{multline*}
  \eta=\mgu(\delta,\theta)=
  \{ u/r(s(w_a,w_a,w_a)), v/s(w_a,w_a), x/s(w_a,w_a),\\ y/s(w_a,w_a),
  z/s(w_a,w_a,w_a), w_b/w_a, w_c/w_a, w_d/w_a, w_e/w_a \} \enspace .
  \end{multline*}
\end{example}

\begin{definition}[Parallel abstract mgu]  
  Given a set of $\omega$-sharing groups $S$ and an idempotent
  substitution $\theta$, the \emph{abstract parallel unification} of
  $S$ and $\theta$ is given by
  \begin{equation}
    \label{eq:mgup}
    \mgup(S,\theta) = \{ \res(\mG) \mid \text{
      $\mathcal G$ is a sharing graph for $S$ and $\theta$}\} \enspace .
  \end{equation}
 This is lifted to the domain $\ShLin^\lp$:
  \begin{equation}
    \label{eq:mgup2}
    \mgup([S]_U,\theta) = [\mgup(S  \cup \{ \multil v \multir \mid 
    v \in \vars(\theta) \setminus U\}, \theta)]_{U \cup \vars(\theta)} 
    \enspace .
  \end{equation}
\end{definition}

It is worth noting that, given any set of $\omega$-sharing groups $S$ and 
substitution $\theta$, there exist many different sharing graphs for 
$S$ and $\theta$.  Each sharing graph yields a resultant sharing 
group which must be  included in  the result of the abstract unification 
operator.  Of course, different sharing  graphs may give the same 
resultant sharing group. The abstract unification operator is defined by 
collecting all the resultant  sharing groups.

\begin{example}
\label{ex:two}
We show another sharing graph for the same $S$ and 
$\theta$ of Example~\ref{ex:psharing}. We omit from the picture the names
of edges and nodes, since they are not relevant here:
 
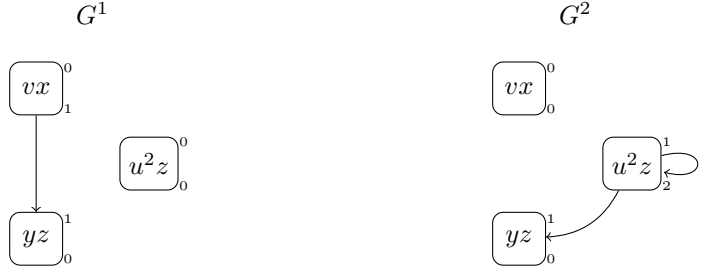
\[
  \begin{tikzpicture}[shgraph]
    \node at (0.5,1) {$G^1$};
    \shnode{0}{1}{} (xv) at (0,0) {$vx$};
    \shnode{1}{0}{} (yz)  at (0,-2) {$yz$} ;
    \shnode{0}{0}{} (u2z) at (1,-1) {$u^2z$};
    \path (xvl) edge  (yz);
  \end{tikzpicture}
  \hspace{4cm}
  \begin{tikzpicture}[shgraph]
    \node at (0.5,1) {$G^2$};
    \shnode{0}{0}{} (xv) at (0,0) {$vx$};
    \shnode{1}{0}{} (yz)  at (0,-2) {$yz$} ;
    \shnode{1}{2}{} (u2z) at (1,-1) {$u^2z$};
    \path (u2z) edge[bend left]  (yz);
    \path (u2z) edge[loop right] (u2z);
  \end{tikzpicture}
  \]
The resultant sharing group is $u^2vxyz$.
\end{example}

It is worth noting that the domain $\Linp$ is not amenable to a direct 
implementation. Actually, it 
may be the  case that, even the mgu of a finite set of $\omega$-sharing groups 
with a 
single-binding substitution generates an infinite set of $\omega$-sharing 
groups (see Example~\ref{ex:infinite} later). However, it is an invaluable 
theoretical device to study the abstract operators for its abstractions, such 
as $\Sharing \times \Lin$ and $\ShLinp$.

\begin{example}
\label{ex:infinite}
It holds that $\mgup(\{ xy\},\{x/y\})= \{ \emptymulti \} \cup \{ x^iy^i  \mid 
i \geq 1 \}$.  Actually, for each $i \geq 1$, the following is a single-layer 
sharing graph:
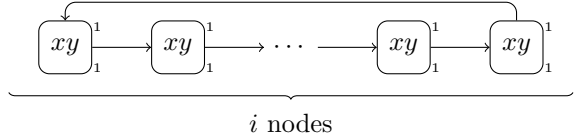
\[
\begin{tikzpicture}[shgraph]
  \shnode 1 1  {} (v1) at (1,0) {$xy$} ;
  \shnode 1 1  {} (v2) at (2,0) {$xy$} ;
  \node (vNull) at (3,0) {$\cdots$} ;
  \shnode 1 1  {} (v3) at (4,0) {$xy$} ;
  \shnode 1 1  {} (v4) at (5,0) {$xy$} ;
  \path (v1) edge (v2);
  \path (v2) edge (vNull);
  \path (vNull) edge (v3);
  \path (v3) edge (v4);
  \draw [rounded corners](v4) -- ++(0,0.6) -- ++(-4,0) -- (v1);
  \draw [decorate,decoration=brace,aspect=3,-] (5.5,-0.6) -- (0.5,-0.6) ;
  \node at (3,-1) {$i$ nodes};
\end{tikzpicture}
\]
\end{example}

\subsection{Coincidence of Parallel and Sequential Abstract Unification}

For concrete substitutions, unification may be performed one binding at a 
time. On an abstract domain, computing one binding at a time generally incurs 
in a loss of precision. However, there are well known domains when this does 
not happen, such as $\Def$ \cite{ArmstrongMSS94} and $\Sharing$. We will 
show that computing one binding at a time does not cause loss of precision on the abstract domain $\Linp$. 

\begin{definition}[Abstract sequential unification]
  Given a set of $\omega$-sharing groups $S$ and an idempotent
  substitution $\theta$, the \emph{abstract sequential unification} 
  of $S$ and $\theta$, denoted by $\mgu_\lp(S,\theta)$, is given by:
  \begin{equation}
  \begin{split}
  \mgu_\lp(S,\epsilon) &=S \\
  \mgu_\lp(S,\{x/t\} \cup \theta)) &=\mgu_\lp(\mgup(S,\{x/t\}),\theta)
  \end{split}
  \end{equation}
\end{definition}
The definition may be lifted to the domain $\ShLin^\lp$ as for $\mgup$. It is 
immediate to check that $\mgup$ and $\mgu_\lp$ are equivalent for 
single-binding substitutions. We will prove that this holds for any 
substitution.

In \cite{AmatoS09sharlin} the abstract sequential unification $\mgu_\lp$ has 
been introduced starting from the definition of a sharing graph for 
single-binding unification. This is essentially a sharing graph with a single 
layer. Hence, it is immediate to check that the
definition of $\mgu_\lp$ given above is the same as the definition of
$\mgu_\lp$ given by \citeN{AmatoS09sharlin}. 

Before introducing the formal proof of coincidence between sequential and 
parallel abstract unification, we try to convey the intuitive idea behind it 
with  an example.

\begin{example}
  Consider again the sharing graph $\mG$ given in Example 
  \ref{ex:psharing} for $S=\{u^2z, uyz, xv, yz \}$ and $\theta=\{x/y,u/r(z)\}$.
  For the sake of conciseness, we can draw $\mG$ with a single picture, 
  omitting the in- and out-degree annotations on the nodes, and with 
  the edges in different styles, according to the layers they come from:
    
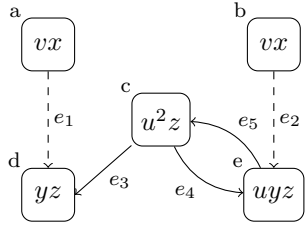
\[
    \begin{tikzpicture}[shgraph]
      \shnode{}{} a (xvl) at (0,0) {$vx$};
      \shnode{}{} d (yz)  at (0,-2) {$yz$} ;
      \shnode{}{} b (xvr) at (2,0) {$vx$};
      \shnode{}{} c (u2z) at (1,-1) {$u^2z$};
      \shnode{}{} e (yuz) at (2,-2) {$uyz$};
      \path (u2z) edge node {$e_3$} (yz.0);
      \path (xvl) edge[dashed] node {$e_1$} (yz);
      \path (xvr) edge[dashed]  node {$e_2$} (yuz);
      \path (u2z) edge[bend right] node[swap]{$e_4$} (yuz);
      \path (yuz) edge[bend right] node[swap]{$e_5$} (u2z);
      \end{tikzpicture}  
    \]
    As we said before, the resultant sharing group of $\mG$ is $u^3v^2x^2y^2z^3$.
    The same sharing group may be obtained by first computing
    $S'=\mgup(S,\{x/y\})$ and later $\mgup(S',\{u/r(z)\})$. 
  Consider the three connected components in the multigraph $G^1$, 
  corresponding to the
  dashed arrows:
  
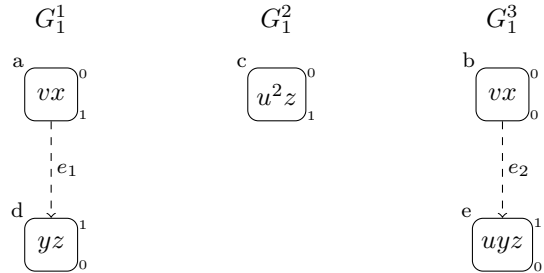
\[
  \begin{tikzpicture}[shgraph]        
        \shnode{0}{1} a (xvl) at (0,0) {$vx$};
        \shnode{1}{0} d (yz)  at (0,-2) {$yz$} ;
        \shnode{0}{0} b (xvr) at (4,0) {$vx$};
        \shnode{0}{1} c (u2z) at (2,0) {$u^2z$};
        \shnode{1}{0} e (yuz) at (4,-2) {$uyz$};
        \node at (0,1) {$G_1^1$};
        \node at (2,1) {$G_1^2$};
        \node at (4,1) {$G_1^3$};
        \path (xvl) edge[dashed] node {$e_1$} (yz);
        \path (xvr) edge[dashed] node {$e_2$} (yuz);
  \end{tikzpicture}
  \]
  Each of them alone may be viewed as a sharing graph with a single 
  layer for the substitution $\{x/y\}$. Therefore, $vxyz$, $u^2z$ and $uvxyz$ 
  are  elements of   $S'$. Now, in the original sharing graph, we  
  collapse these   connected   components:
  
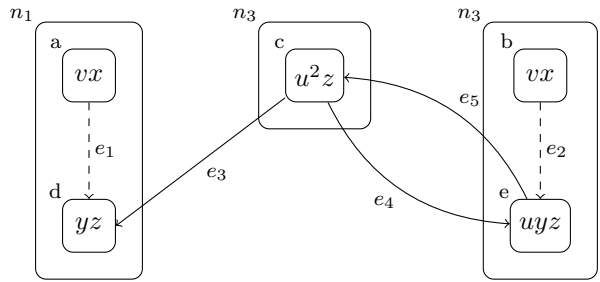
\[
  \begin{tikzpicture}[shgraph]
     \shnode{}{} a (xvl) at (0,0) {$vx$};
     \shnode{}{} d (yz)  at (0,-2) {$yz$} ;
     \shnode{}{} b (xvr) at (4,0) {$vx$};
     \shnode{}{} c (u2z) at (2,0) {$u^2z$};
     \shnode{}{} e (yuz) at (4,-2) {$uyz$};
     \node[group,fit={(xvl)(yz)}, inner sep=10pt, label=north west:$n_1$] (vxyz) {};
     \node[group,fit={(u2z)}, inner sep=10pt, label=north west:$n_3$] (u2zbig) {};
     \node[group,fit={(xvr) (yuz)}, inner sep=10pt,  label=north west:$n_3$] (uvxyz) {};
     \path (xvl) edge[dashed] node {$e_1$} (yz);
     \path (xvr) edge[dashed] node {$e_2$} (yuz);
     \path (u2z) edge node {$e_3$} (yz.0);
     \path (u2z) edge[bend right] node[swap]{$e_4$} (yuz);
     \path (yuz) edge[bend right] node[swap]{$e_5$} (u2z);
  \end{tikzpicture}
  \]
  and we get
  
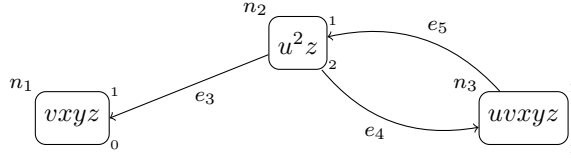
\[
    \begin{tikzpicture}[shgraph]
     \shnode{1}{0}{$n_1$} (vxyz)  at (0,-1) {$vxyz$} ;
     \shnode{1}{2}{$n_2$} (u2z) at (2,0) {$u^2z$};
     \shnode{1}{1}{$n_3$} (uvxyz) at  (4,-1) {$uvxyz$};
     \path (u2z) edge node {$e_3$} (vxyz.0);
     \path (u2z) edge[bend right] node[swap]{$e_4$} (uvxyz);
     \path (uvxyz) edge[bend right] node[swap]{$e_5$} (u2z);
     \end{tikzpicture}
  \]
  which is a sharing graph for $S'$ and $\{u/r(z)\}$.  Note that, in this new
  sharing graph, the nodes correspond to the connected components of $G^1$ 
  and the edges are the same as in the original $G^2$, but with different 
  source and target. The edge
  $e_3$ from $c$ to $d$ is now an edge from $n_2$ to $n_1$, since $d$ is in the
  first connected component and $c$ in the second one.  We obtain, as 
  expected, that $u^3v^2x^2y^2z^3 \in \mgu_\lp(S',\{u/r(z)\})$.
\end{example}

\begin{example}
\label{ex:three}
We now show an example of the converse, i.e., how to move from sequential to   
parallel unification.  Assume 
$S=\{vx, u^2, uvx, uvyz, uy\}$ and $\theta=\{x/y,  u/s(z,s(z,v)) \}$. The 
following are single-layer sharing graphs for $S$ and $\{x/y\}$:
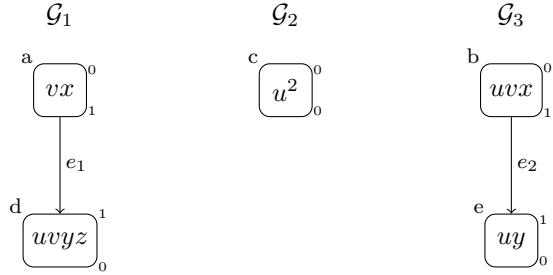
  \[
  \begin{tikzpicture}[shgraph]        
     \shnode{0}{1} a (xvl) at (0,0) {$vx$};
     \shnode{1}{0} d (yz)  at (0,-2) {$uvyz$} ;
     \shnode{0}{1} b (xuv) at (4,0) {$uvx$};
     \shnode{0}{0} c (u2z) at (2,0) {$u^2$};
     \shnode{1}{0} e (yu) at (4,-2) {$uy$};
     \node at (0,1) {$\mG_1$};
     \node at (2,1) {$\mG_2$};
     \node at (4,1) {$\mG_3$};
     \path (xvl) edge node{$e_1$} (yz);
     \path (xuv) edge node{$e_2$} (yu);
  \end{tikzpicture}
  \]
Note that we have chosen disjoint sets of nodes $N_{\mG_1} = \{a,d\}$, 
$N_{\mG_2}=\{c\}$ and $N_{\mG_3}=\{b,e\}$, and disjoint sets of edges 
$E_{\mG_1} = \{e_1\}$, $E_{\mG_2}=\{\}$ and
$E_{\mG_3}=\{e_2\}$. By definition, the corresponding resultant 
sharing groups, i.e., $uv^2xyz$, $u^2$ and $u^2vxy$ are elements of 
$S'=\mgu_\lp(S,\{x/y\})$. 
Now consider the following sharing graph $\mG$ for $S'$ and the 
binding 
$u/s(z,s(z,v))$:
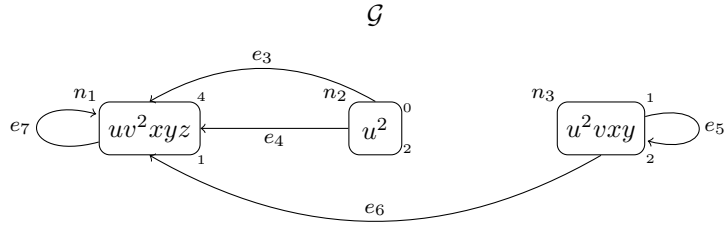
  \begin{equation}
  \label{eq:psharingtwo}
  \begin{tikzpicture}[shgraph]        
    \node at (2,1.5) {$\mG$};
    \shnode{4}{1}{$n_1$} (vxyz) at (0,0) {$uv^2xyz$};      
    \shnode{0}{2}{$n_2$} (u2z) at (2,0) {$u^2$};
    \shnode{1}{2}{$n_3$} (u2vxy) at (4,0) {$u^2vxy$};
    \path (vxyz) edge[loop left] node{$e_7$} (vxyz);
    \path (u2z.north) edge[bend right] node[swap]{$e_3$} (vxyz.north) ;
    \path (u2z) edge node{$e_4$} (vxyz) ;
    \path (u2vxy) edge[loop right] node{$e_5$} (u2vxy) ;
    \path (u2vxy.south) edge[bend left] node[swap]{$e_6$} (vxyz.south) ;
   \end{tikzpicture}
  \end{equation}
We need to build a sharing graph for $S$ and $\{x/y, u/s(z,s(z,v))\}$ 
from these pieces. 
The idea is to replace, in the graph $\mG$, the nodes $n_1$, $n_2$ and $n_3$ 
with the graphs $\mG_1$, $\mG_2$ and $\mG_3$ respectively:
  
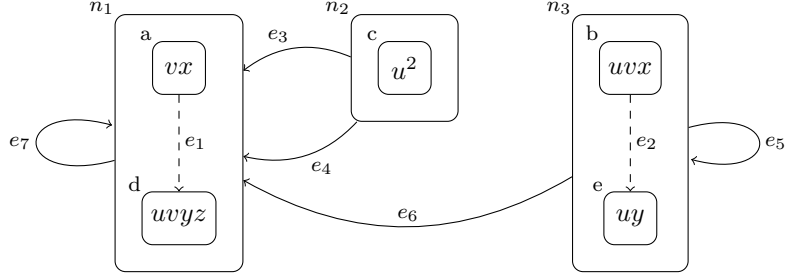
\[
  \begin{tikzpicture}[shgraph]   
    \shnode{}{} a (xvl) at (0,0) {$vx$};
    \shnode{}{} d (yz)  at (0,-2) {$uvyz$} ;
    \shnode{}{} b (xuv) at (4,0) {$uvx$};
    \shnode{}{} c (u2zbis) at (2,0) {$u^2$};
    \shnode{}{} e (yu) at (4,-2) {$uy$};     
    \path (xvl) edge[dashed] node{$e_1$} (yz);
    \path (xuv) edge[dashed] node{$e_2$} (yu);
    \shnode{}{}{$n_1$}[fit={(xvl)(yz)}, inner sep=10pt] (vxyz) 
    {};      
    \shnode{}{}{$n_2$}[fit={(u2zbis)}, inner sep=10pt] (u2z) {};
    \shnode{}{}{$n_3$}[fit={(xuv)(yu)}, inner sep=10pt] (u2vxy) {};
    \path (u2z) edge[bend right] node[swap]{$e_3$} (vxyz) ;
    \path (u2z) edge[bend left] node{$e_4$} (vxyz) ;
    \path (u2vxy) edge[loop right] node{$e_5$} (u2vxy) ;
    \path (u2vxy) edge[bend left] node[swap]{$e_6$} (vxyz) ;
    \path (vxyz) edge[loop left]  node{$e_7$} (vxyz);
   \end{tikzpicture}
  \]
For each edge in $\mG$ we need to specify its target and source as a node in 
$\{a,b,c,d,e\}$, since giving only the connected component is not enough. For 
example, the target of $e_1$ should be either  $a$ or $d$. We may choose the 
targets freely, subject
to the conditions on the in-/out- degree of nodes. Since  
$\chi(vx,s(z,s(z,v)))=1$ and $\chi(uvyz,s(z,s(z,v)))=3$, among $e_3$, $e_4$, 
$e_6$  and $e_7$, three edges should be targeted  at  $uvyz$ and one should be 
targeted at $vx$. Among the many others, this is a possible sharing graph, 
where the different layers are depicted trough different line styles:
  
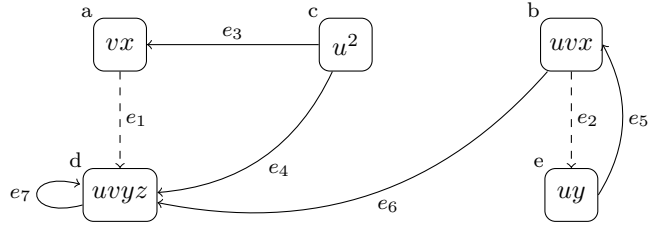
\[
  \begin{tikzpicture}[shgraph]   
    \shnode{}{} a (xvl) at (0,0) {$vx$};
    \shnode{}{} d (yz)  at (0,-2) {$uvyz$} ;
    \shnode{}{} b (xuv) at (4,0) {$uvx$};
    \shnode{}{} c (u2z) at (2,0) {$u^2$};
    \shnode{}{} e (yu) at (4,-2) {$uy$};     
    \path (xvl) edge[dashed] node{$e_1$} (yz);
    \path (xuv) edge[dashed] node{$e_2$} (yu);
    \path (u2z) edge node[swap]{$e_3$} (xv) ;
    \path (u2z) edge[bend left] node{$e_4$} (yz) ;
    \path (xuv) edge[bend left] node{$e_6$} (yz) ;
    \path(yu.east) edge[bend right] node[swap]{$e_5$} (xuv.east) ;
    \path (yz) edge[loop left]  node{$e_7$} (yz);
   \end{tikzpicture}
  \]
Note  that the self loop on the node $n_3$ has become an edge from $uy$ to 
$uvx$. 
\end{example}

The ideas presented in the previous examples are formalized in the following 
result.
\begin{lemma}
  Given a set of $\omega$-sharing groups $S$ and an idempotent
  substitution $\theta$, we have that
\label{lem:mgucomp}
  $\mgup(S,\{x_1/t_1\} \cup \theta ) = \mgup(\mgup(S,\{x_1/t_1\}),\theta)$.
\end{lemma}

\begin{proof}
  If $\theta=\epsilon$ the result easily follows  since
  $\mgup(S,\epsilon)= S$.  In the case $\theta \neq \epsilon$, we separately
  prove the two sides of the equality.

  \textbf{First part: $\subseteq$ inclusion.}   Let $B\in
  \mgup(S,\{x_1/t_1\} \cup \theta)$. We want to prove that $B \in
  \mgup(\mgup(S,\{x_1/t_1\}),\theta)$. To this aim, we will provide a 
   sharing graph $\mathcal \mG'$ for 
  $\mgup(S,\{x_1/t_1\})$ and $\theta$ such that $\res(\mG')=B$.
  
  Let $\theta=\{x_2/t_2,\ldots,x_p/t_p\}$. By definition, there exists a
   sharing graph $\mathcal G=\{G^i\}_{i \in [1,p]}$ such that
  $B=\res(\mG)$. We decompose $G^1$ into its
  connected components $G^1_1, \ldots, G^1_k$. Note that each $G^1_j$,
  labeled with the obvious restriction of $l_{\mG}$, is a sharing
  graph for $S$ and $\{x_1/t_1\}$, therefore $\res(G^1_j) \in
  \mgup(S,\{x_1/t_1\})$.
  
  We now show a sharing graph $\mathcal \mG'$ for 
  $\mgup(S,\{x_1/t_1\})$ and $\theta$ and prove that $\res(\mG')=B$.
  For any $i \in [2,p]$, let $G_i$ be the multigraph obtained
  from $G^i$ by collapsing each of the connected components $G_1^1,
  \ldots G^1_k$ to a single node. Formally:
  \begin{itemize}
  \item $N_{G_i}=\{ 1, \ldots, k \}$;
  \item $E_{G_i}=E_{G^i}$;
  \item $\src_{G_i}(e)=j$ iff $\src_{G^i}(e) \in G^1_j$;
  \item symmetrically for $\tgt_{G_i}$.
  \end{itemize}
  
  We want to prove that $\mathcal \mG'=\{G_i\}_{i \in [2,p]}$, 
  endowed with the labeling function $l_{\mG'}(j)=\res(G^1_j)$, is a
   sharing graph for $\mgup(S,\{x_1/t_1\})$ and $\theta$.
  By definition of sharing graph, we need to check that: first, the conditions 
  on the out-degree and the in-degree hold for each node; second, the sets of 
    edges are pairwise disjoint; third, the flattening is connected.
 
  \textbf{First condition.} We now show that the conditions on the out-degree 
  and the in-degree of the nodes hold.
  Given any node $j \in [1,k]$ we have that the
  out-degree of $j$ in $G_i$ is
  \[
  \begin{split}
    &\card{\{ e \in E_{G_i} \mid \src_{G_i}(e)=j \}}
      = \card{\{ e \in E_{G^i} \mid \src_{G^i}(e) \in G^1_j \}}\\
    &= \sum_{n \in N_{G^1_j}} \card{\{ e \in E_{G^i} \mid
      \src_{G^i}(e)=n \}}
     = \sum_{n \in N_{G^1_j}} \chi(l_{\mG}(n),x_i)\\
    &= \sum_{n \in N_{G^1_j}} \sum_{v \in \var} l_{\mG}(n)(v) \cdot
    \occ(v,x_i)
     = \sum_{v \in \var} \sum_{n \in N_{G^1_j}} l_{\mG}(n)(v) \cdot
    \occ(v,x_i)\\ 
    &= \sum_{v \in \var} \res(G^1_j)(v) \cdot \occ(v,x_i)
     = \sum_{v \in \var} l_{\mG'}(j)(v) \cdot \occ(v,x_i)\\ 
    &=  \chi(l_{\mG'}(j),x_i)\enspace .
  \end{split}
  \]
  Symmetrically, we have that the in-degree of $j$ in $G_i$ is
  $\chi(l_{\mG'}(j),t_i)$. 

  \textbf{Second condition.} It is immediate to check that the sets of edges 
  $E_{G_i}$ are pairwise disjoint. 

  \textbf{Third condition.} We prove that $\flatten{\mG'}$ is connected. 
  Assume that we want to find a path from $i$ to $j$.  Since
  $\flatten{\mathcal G}$ is connected, there is a path $\pi$ from
  some $n_1 \in N_{G^1_i}$ to some $n_2 \in N_{G^1_j}$. A path from
  $i$ to $j$ may be obtained in two steps:
  \begin{enumerate}
  \item by replacing each node $n$ in $\pi$ with $\bar{n}$ where
    $\bar{n}$ is the unique $m \in [1,k]$ such that $m \in N_{G^1_m}$;
  \item by replacing each subsequence $\bar{n} \bar{n}$ with a single
    node $\bar{n}$. Such a situation may arise when $\pi$ contains
    the subsequence $n m$ with $n \ra m \in G^1_q$ for some $q$. The
    corresponding edge $q \ra q$ may not exists in $\mG'$, but being a
    self-loop it may be deleted.
  \end{enumerate}
  
  Finally, we need to show that $\res(\mG')=B$. It
  is easy to check that $\res(\mG') = \bigmultisum_{i
    \in [1,k]} l_{\mG'}(i)=\bigmultisum_{i \in [1,k]} \res(G^1_i)=
  \bigmultisum_{i \in [1,k]} \bigmultisum_{n \in N_{G^1_i}} l_{\mG}(n)=
  \bigmultisum_{n \in N_{\mG}} l_{\mG}(n)=\res(\mG)$.
  
  \textbf{Second part: $\supseteq$ inclusion.} Let $S'=\mgup(S,\{x_1/t_1\})$ 
  and $B\in \mgup( S',\theta )$ where $\theta=\{x_2/t_2, \ldots, x_p/t_p\}$. 
  We show that there exists a sharing graph 
$\mG'$ for $S$ and $\{x_1/t_1, \ldots, x_p/t_p\}$ such that $\res(\mG')=B$.

  By definition, there is a sharing graph $\mathcal G=\{ G^i
  \}_{i \in [2,p]}$ for $S'$ and $\theta$ such that $\res(\mG)=B$. Since
  $S'=\mgup(S,\{x_1/t_1\})$, for each node $k \in N_{\mG}$ we have a
  sharing graph $G_k$ such that $\res(G_k) = l_{\mG}(k)$.  Without
  loss of generality, we may choose these graphs in such a way that
  the sets $N_{G_k}$ are pairwise disjoint and disjoint from
  $N_{\mG}$.
  
  For each multigraph $G^i$, with $i \in [2,p]$, we build a new multigraph 
  $\bar G^i$ obtained by replacing each node $k$ in $G^i$ with the set of 
  nodes of the generating graph $G_k$. Then, we pack the $\bar G^i$'s and 
  $G_k$'s into a sharing graph $\mG$. Formally, $\mG = \{ \bar G^i \}_{i \in 
  [1,p]}$ such that:
  \begin{itemize}
  \item $N_{\mG}=\bigcup_{k \in N_{\mG}} N_{G_k}$;
  \item $\bar G^1$ is the  union of the graphs $G_k$;  
  \item for $i \in [2,p]$, $E_{\bar G^i}=E_{G^i}$ ;
  \item for $i \in [2,p]$, $\src_{\bar G^i}$ is chosen freely, subject to the 
  following conditions:
  \begin{itemize}
  \item if $\src_{G^i}(e)=k$ then $\src_{\bar G^i}(e)$ is a node in $N_{G_k}$;
  \item the out-degree of each node $n$ in $\bar G^i$ is $\chi(l(n),x_i)$.   
  \end{itemize} 
   This is always possible since $l_{\mG}(k) = \res(G_k) = \bigmultisum_{n \in 
  N_{G_k}} l_{G_k}(n)$ and therefore $\chi(l_{\mG}(k),x_i)=\sum_{n \in      
  N_{G_k}} \chi(l_{G_k}(n), x_i)$.  Symmetrically for $\tgt_{\bar G^i}$.
  \item the labeling function $l: N_{\mG} \ra \mwp(\var)$
    is the disjoint union of all the $l_{G_k}$. Namely, $l(n)=l_{G_k}(n)$ iff 
    $n \in N_{G_k}$.
  \end{itemize}

  We now want to prove that $\mG'$ is a sharing graph for $\{x_1/t_1,
  \ldots, x_p/t_p\}$ and $S$. The only thing we need to prove is 
  that $\flatten{\mG'}$ is connected (the other conditions hold by 
  construction).
  
  Assume that there is an edge $i \ra j$ in $G^k$, and consider nodes $n_i
  \in N_{G_i}$ and $n_j \in N_{G_j}$. We prove that there is a path in
  $\flatten{\mG'}$ from $n_i$ to $n_j$. Actually, there is in
  $\bar{G}^k$ at least an edge $m_i \ra m_j$ from a node $m_i \in
  N_{G_i}$ to $m_j \in N_{G_j}$. Since $G_i$ and $G_j$ are connected,
  there are in $\bar{G}^1$ two paths $\pi: n_i \ra m_i$ and $\pi': m_j
  \ra n_j$. Therefore $\pi \pi'$ is a path in $\flatten{\mG'}$ from
  $n_i$ to $n_j$. 
  
  Now, given two generic nodes $n_i, n_j$ where $n_i \in N_{G_i}$ and
  $n_j \in N_{G_j}$, we know there is a path $\pi$ in
  $\flatten{\mG}$ from $i$ to $j$. Applying the result of the
  previous paragraph to each edge in $\pi$, we immediately get that
  $n_i$ and $n_j$ are connected.
     
  Finally it is easy to check that $\res(\mG')=B$
  and this concludes the proof of the theorem.
\end{proof}

By exploiting the previous lemma, it is now a trivial task to show that 
parallel and sequential unification compute the same result.

\begin{theorem}\label{th:coincide}
  The abstract operators $\mgu_\lp$ and $\mgup$ coincide.
\end{theorem}  
\begin{proof}
  The proof is by induction on the number of bindings in $\theta$.  Clearly
  $\mgu_\lp(S,\epsilon) = S =\mgup(S,\epsilon)$. Assume that
  $\mgu_\lp(S,\theta) = \mgup(S,\theta)$ for each $S$.  It follows
  that
  \[
  \begin{array}{ll}
     \mgu_\lp(S,\{x/t\} \cup \theta) \\
    = \mgu_\omega(\mgup(S,\{x/t\}),\theta) & 
    \text{[by definition of $\mgu_\omega$]}\\
    = \mgup(\mgup(S,\{x/t\}),\theta) \quad & 
    \text{[by induction hypothesis]}\\
%    = & \mgup(\mgup(S,\{x/t\}),\theta) & 
%    \text{$\mgu_\lp$ and $\mgup$ coincide for a single-binding}\\
    = \mgup(S,\{x/t\} \cup \theta) & 
    \text{[by Lemma \ref{lem:mgucomp}]}
  \end{array}
  \]
  and this proves the theorem.
\end{proof}

\subsection{Optimality of Abstract Unification}

An immediate consequence of Theorem \ref{th:coincide} is that $\mgup$ is 
correct, since it coincides with $\mgu_\lp$ which has been proved correct in 
\cite{AmatoS09sharlin}. We now want to prove that it is optimal.  First, we 
prove optimality in the special case of $\mgup([S]_U,\theta)$ with 
$\vars(\theta) \subseteq U$. Next, we extend this result to the general case.

In the Example~\ref{ex:psharing2}, we have already shown how to build
a substitution $\delta$ which mimics the effect of a sharing graph. 
We now give another example, introducing the terminology to be used in the
proof of optimality to come.

\begin{example}\label{ex:optimality}
We refer to Example \ref{ex:three}. Let $U=\{u,v,x,y,z\}$ be the set 
of variables of interest. We show how to build a substitution  
$\delta$ such that  $[S]_U \rightslice [\delta]_U$ and $u^3v^2x^2y^2z^3 \in 
\alpha_\lp(\mgu([\delta]_U,\theta))$. 

For each node $n \in \{a,b,c,d,e\}$ of the sharing graph in 
\eqref{eq:psharingtwo}, we consider a different fresh  variable $w_n$. For any 
variable $\nu \in U \setminus \dom(\theta) = \{ v, y, z \}$, we define 
$\delta(\nu)$ as the following term of arity $\sum_n 
l_\mG(n)(\nu)$:
\[
\delta(\nu) =  r(\underbrace{w_a,\ldots,w_a}_{\text{$l_{\mG}(a)(\nu)$
      times}}, \underbrace{w_b,\ldots,w_b}_{\text{$l_{\mG}(b)(\nu)$
      times}}, \ldots,
  \underbrace{w_e,\ldots,w_e}_{\text{$l_{\mG}(e)(\nu)$ times}}) \enspace .
  \]
Since $l_{\mG}(n)$ is the label of the node $n$ and $l_{\mG}(n)(\nu)$ is 
the multiplicity of 
$\nu$ in such a label, we have:
\[
\delta(v) = r(w_a,w_b,w_d) \qquad \delta(y) = r(w_d,w_e) \qquad \delta(z) = 
r(w_d) \enspace .
\]
For the variables in $\dom(\theta)=\{u,x\}$ we define $\delta$ in a different 
way. Consider the first layer of the sharing graph, corresponding to the 
binding $x/y$, and denote by $f^1$ an injective map from occurrences of 
variables $w_n$ in $\delta(\theta(x))$ to edges targeted  at $n$. 
In  this case, we have $\delta(\theta(x))=r(w_d,w_e)$ and $f^1 = \{ \xi_d 
\mapsto e_1, \xi_e \mapsto e_2 \}$, where $\xi_d = 1$ and $\xi_e = 2$ are 
the positions of $w_d$  and $w_e$ in $\delta(\theta(x))$. 

Analogously, we define $f^2$ for the binding $u/s(z,s(z,v))$. In this 
case $\delta(\theta(u))=s(r(w_d), s(r(w_d), r(w_a,w_b,w_d))$ and a possible 
$f^2$ is $\{ 1\cdot 1 \mapsto e_7, 2 \cdot 1 \cdot 1 \mapsto e_4, 2 \cdot 2 
\cdot 1 \mapsto e_3, 2 \cdot 2 \cdot 2 \mapsto e_5, 2 \cdot 2 \cdot 3 \mapsto 
e_6 \}$. In this case, other values for $f^2$ are possible: we could exchange 
the assignments for $1\cdot 1$, $2 \cdot 1 \cdot 1$ and  $2 \cdot 2 \cdot 3$ 
freely.

We now define $\delta(x)=f^1(\delta(\theta(x)))=r(w_a,w_b)$. Here we denote 
with $f^1(t)$ the result of replacing, in the term $t$, the variable in 
position $\xi$ with the variable associated to the source of $f^1(\xi)$. 
Analogously we define $\delta(u)=f^2(\delta(\theta(u)))=s(r(w_d), s(r(w_c), 
r(w_c,w_e,w_b))$. 

Note that the terms $\delta(x)$ and $\delta(u)$ are obtained by replacing in 
$\theta(x)$ and $\theta(u)$ each occurrence $\xi$ of variable $\nu$ with a 
variant of $\delta(\nu)$. We call $\delta^1_{1}(y)$ the term which replaces 
$y$ in position $1$ of  $\theta(x)$, \ie $r(w_d,w_e)$. Analogously, we 
define $\delta^2_{1}(y) = r(w_d)$, $\delta^2_{2 \cdot 1}(y) = r(w_c)$, and 
$\delta^2_{2 \cdot 2}(v) = r(w_c,w_e,w_b)$ for the replacements in 
$\theta(u)$ of  $y$ in position $1$, $y$ in position $2 \cdot 1$ 
and $v$ in position $2 \cdot 2$ respectively. This terminology will be used in 
the proof.

We have that $[S]_U \rightslice [\delta]_U$ and
$\mgu(\delta, \theta)= \mgu(\{x=y, u=s(z,s(z,v)), u=s(r(w_d), s(r(w_c), 
r(w_c,w_e,w_b)), v=r(w_a,w_b,w_d), x=r(w_a,w_b), y=r(w_d,w_e), z=r(w_d)\})$ is
\[
    \begin{split}
       \mgu(\delta,\theta)
      &= \theta \circ \mgu(\{ v=r(w_a,w_b,w_d), y=r(w_d,w_e), 
      z=r(w_d) \} \cup\\
      &\phantom{{}={}} \{ y=r(w_a,w_b), s(z,s(z,v)) = s(r(w_d), s(r(w_c), 
      r(w_c,w_e,w_b))) \})\\
      &= \theta \circ \mgu(\{ v=r(w_a,w_b,w_d), y=r(w_d,w_e), 
      z=r(w_d) \} \cup \\
      & \phantom{{}={}} \{  y=r(w_a,w_b), z=r(w_d), z=r(w_c), v=r(w_c,w_e,w_b) 
      \}\\
      &=\theta \circ \delta|_{U \setminus \dom(\theta)} \circ 
      \mgu(\{r(w_d,w_e)=r(w_a,w_b), r(w_d)=r(w_d), \\
      & \phantom{{}={}} r(w_d)=r(w_c), 
      r(w_a,w_b,w_d)=r(w_c,w_e,w_b) \}) \\
      &=\theta \circ \delta|_{U \setminus \dom(\theta)} \circ 
        \mgu(\{w_d=w_a, w_e=w_b, w_d=w_d, \\
      & \phantom{{}={}} w_d=w_c, w_a=w_c, w_b=w_e, w_d=w_b \})      
      \enspace .
    \end{split}
\]
In the last formula, we have an equation $w_n=w_m$ for each edge $n \ra m$ in 
the sharing graph. Since the graph is connected, we pick a variable, say it is 
$w_d$, and we solve the set of equations w.r.t.~that variable, obtaining:
\[
\begin{split}
\mgu(\delta, \theta) &= \theta \circ \delta|_{U \setminus \dom(\theta)} \circ 
\{ w_a/w_d, w_b/w_d, w_c/w_d, w_e/w_d \}\\
&= \{ u/s(r(w_d),s(r(w_d),r(w_d,w_d,w_d)), v/r(w_d,w_d,w_d),\\
&\phantom{{}={}} x/r(w_d,w_d), y/r(w_d,w_d), z/r(w_d)\}
\enspace.
\end{split}
\]
We get $\alpha_\lp(\mgu([\delta]_U,\theta)) = [u^5v^3x^2r^2z]_U$,
where $u^5v^3x^2r^2z= \mgu(\delta,\theta)^{-1}(w_d) |_{U}$.
\end{example}

The above example shows how to find a substitution whose fresh 
variables are aliased according to the arrows in a sharing 
graph. The same idea is exploited in the next theorem for proving the 
optimality of the abstract unification operator $\mgup([S]_U,\theta)$.

%For any 
%$\omega$-sharing group $B \in \mgup([S]_U,\theta)$, we provide a substitution 
%$\delta$ obtained as in Example \ref{ex:optimality}, such that $[S]_U$ 
%approximates $[\delta]_U$ and 
%$B \in  \alpha_\lp(\mgu([\delta]_U,\theta))$. 

\begin{theorem}
  \label{th:lpopt}
The parallel unification $\mgup([S]_U,\theta)$ is optimal \wrt
$\mgu$, under the assumption that $\vars(\theta) \subseteq U$, that is:
\[
\forall B\in\mgup ([S]_U,\theta)~\exists [\delta]_U \in \Isubst_{\sim}.~[S]_U \rightslice 
[\delta]_U \text{ and } B \in \alpha_\lp(\mgu([\delta]_U,\theta)) \enspace .
 \]
\end{theorem}
\begin{proof}
  Let $\theta =\{x_1/t_1, \ldots, x_p/t_p\}$ and $B \in
  \mgup([S]_U,\theta)$. By definition of $\mgup$, there exists a 
  sharing graph $\mathcal G=\{G^i\}_{i\in [1,p]}$ such that $B =
  \res(\mG)$. Let $N_{\mG}=\{ n_1, \ldots, n_k \}$.
    We want to define a substitution $\delta$ such that $[S]_U
  \rightslice [\delta]_U$ and $B \in
  \alpha_\lp(\mgu([\delta]_U,\theta))$.
If $B=\emptymulti$ this is trivial, just take $\delta=\epsilon$, hence we 
assume that $B \neq \emptymulti$.   
The structure of the proof is as follows: first, we define a substitution 
$\delta$ which unifies with $\theta$; second, we show that $\delta$ is 
approximated by $[S]_U$, namely, $[S]_U  \rightslice [\delta]_U$; third, we 
show that $B \in
  \alpha_\lp(\mgu([\delta]_U,\theta))$.

  \textbf{First part.} We define a substitution $\delta$ which unifies 
  with $\theta$. For each node $n \in
  N_{\mG}$ we consider a fresh variable $w_{n}$ and we denote by $W$
  the set of all these new variables.
 
  For any $y \in U \setminus \dom(\theta)$ we define as $\delta(y)$ the 
  term of arity $\sum_{n\in N_{\mG}}l_{\mG}(n)(y)$ given by:
  \[
  \delta(y)= r(\underbrace{w_{n_1},\ldots,w_{n_1}}_{\text{$l_{\mG}(n_1)(y)$
      times}}, \underbrace{w_{n_2},\ldots,w_{n_2}}_{\text{$l_{\mG}(n_2)(y)$
      times}}, \ldots,
  \underbrace{w_{n_k},\ldots,w_{n_k}}_{\text{$l_{\mG}(n_k)(y)$ times}})
  \enspace .
  \]
  
  For any $x_i \in \dom(\theta)$, consider an
  injective function $f^i$ which maps each occurrence of a variable $w_n$ in 
  $\delta(t_i)$ to an edge in $E_{G^i}$ targeted at $n$. Note that the map 
  exists 
  since the number of occurrences of $w_n$ in $\delta(t_i)$ is exactly 
  $\sum_{y 
  \in t_i} \occ(y,t_i) \cdot \occ(w_n,\delta(y)) =  \sum_{y 
  \in t_i} \occ(y,t_i) \cdot l_\mG(n)(y) = \chi(l_{\mG}(n),t_i)$ which 
  is the in-degree of $n$. Then, we define $\delta(x_i)$ as $f^i(\delta(t_i))$ 
  where $f^i(\delta(t_i))$ is the result of replacing, in $\delta(t_i)$, the 
  variable in position $\xi$ with the variable associated to the source 
  of $f^i(\xi)$.   

  The image of $f^i$ is the set of all the edges in 
  $G^i$. Given an edge $e: n_1 \ra n_2$, the sharing group associated to 
  $n_2$ should contain at least a variable $y \in t_i$, hence $w_{n_2}$
  will occur in $\delta(y)$ and $e$ will be $f^i(\xi)$ for some occurrence
  $\xi$ of $w_{n_2}$ in $\delta(t_i)$.

  \textbf{Second part.} Now we show that $[S]_U \rightslice [\delta]_U$. We need to consider
  all the variables $v \in \var$ and check that $\delta^{-1}(v)|_U \in
  S$. We distinguish several cases:
  \begin{itemize}
  \item let us choose as $v$ the variable $w_n$ for some $n \in N_\mG$. By
    construction, for each $y \in U \setminus \dom(\theta)$, we have 
    that $\occ(w_n,\delta(y))=l_{\mG}(n)(y)$.   Since $\mG$ is a  
    sharing graph, for any $x_i \in \dom(\theta)$ 
    there are $l_{\mG}(n)(x_i)$ edges in $E^i$ departing from $n$. They are
    all in the image of $f^i$, hence $\occ(w_n,\delta(x_i)) = 
    l_{\mG}(n)(x_i)$.We obtain the
    required result which is $\delta^{-1}(w_n)|_U=l_{\mG}(n) \in S$.
  \item if we choose a variable $v \in U$ then $v \in \dom(\delta)$
    and $\delta^{-1}(v)=\emptymulti \in S$;
  \item finally, if $v \notin U \cup W$, then $\delta^{-1}(v)=\multil
    v \multir$ and $\delta^{-1}(v)|_{U}=\emptymulti \in S$.
  \end{itemize}
  
  \textbf{Third part.} We now show that $B \in 
  \alpha_\lp(\mgu([\delta]_U,\theta))$.   
  Note that $\delta(x_i)$ is obtained by replacing, in $t_i$, each occurrence
  $\xi$ of a variable $y \in \rng(\theta)$ with a variant of $\delta(y)$. We 
  denote this variant by $\delta^i_{\xi}(y) = f^i(\delta(t_i))(\xi)$.
  By definition of $\mgu$ over $\Isubst_\sim$, we have that
  $\mgu([\delta]_U,\theta)=[\mgu(\delta,\theta)]_U$. We obtain:
  \begin{equation}
    \label{eq:opt-proof1}
    \begin{split}
      \eta &= \mgu(\delta,\theta)\\
       &= \theta \circ \mgu \bigl(
                  \eq(\delta|_{U \setminus \dom(\theta)}) \cup 
             \{ \theta(x_i)=\delta(x_i)) \mid i \in [1,p] \}
                  \bigr)\\
      &= \theta \circ \mgu \bigl(
            \eq(\delta|_{U \setminus \dom(\theta)}) \cup 
       \{ t_i=f^i(\delta(t_i)) \mid i \in [1,p]  \}
            \bigr)\\
      &= \theta \circ \mgu (\eq(\delta|_{U \setminus \dom(\theta)}) \cup {}
 \{ y=\delta^i_\xi(y) \mid i \in [1,p], 
                t_i(\xi)=y\}\bigr)\\
      &= \theta \circ \delta_{U \setminus \dom(\theta)} \circ 
      \mgu\bigl(\{\delta(y)=\delta^i_\xi(y) \mid i \in [1,p], t_i(\xi)=y 
      \}\bigr) \enspace .
    \end{split}
  \end{equation}
  
  The set of equations $F=\{\delta(y)=\delta^i_\xi(y) \mid i \in [1,p], 
  t_i(\xi)=y \}$ has a solution, given by aliasing
  some variables. We show that, for any edge $e: n 
  \ra m \in E_{G^i}$, it follows from $F$ that $w_n=w_m$. Since the image
  of $f^i$ is the set of all the edges in $G^i$, there is an occurrence $\xi$ 
  of $w_m$ in $\delta(t_i)$ such that $f^i(\xi)=e$. Occurrence $\xi$ 
  may be written as $\xi' \cdot \xi''$ where $\xi'$ is an occurrence of a 
  variable $y \in t_i$ and $\xi''$ is an occurrence of $w_m$ in 
  $\delta(y)$. Therefore $\delta(y)=\delta^i_{\xi'}(y) \in F$, and from
  this follows $w_m = \delta(y)(\xi'')$ is unified with $w_n = 
  \delta^i_{\xi'}(y)(\xi'')$.
  
  Since this holds for any edge in $E_{G^i}$ and for any $i \in
  [1,p]$, it follows that for any edge $n\ra m \in
  E_{\flatten{\mG}}$ the equation $w_m=w_n$ is entailed by $F$. We
  know that $\flatten{\mG}$ is connected, hence for any $n,m\in N_{\mG}$,
  the set of equations in $F$ implies $w_n=w_m$.  We choose a
  particular node $\bar{n} \in N_{\mG}$ and, for what we said before, we
  have $\mgu(F)=\{ w_{n} / w_{\bar{n}} \mid n \in N_{\mG}, n \neq \bar n \}$. 
  We show that
  $\eta^{-1}(w_{\bar{n}})|_{U}=B$.
  \[
    \begin{split}           
      &\eta^{-1}(w_{\bar{n}})|_{U}\\
      &= \theta^{-1}(\delta^{-1}|_{U \setminus \dom(\theta)}(\multil 
      w_{n_1},\ldots,w_{n_k}\multir ))|_{U}=\\
      &= \theta^{-1}(\multil w_{n_1},\ldots,w_{n_k} \multir \multisum
      \lambda y \in U \setminus \dom(\theta).
      \sum_{n \in N_{\mG}} l_{\mG}(n)(y))|_{U}=\\
      &= \lambda y \in U \setminus \dom(\theta). \sum_{n \in N_{\mG}}
      l_{\mG}(n)(y) \ \multisum \\
      & \phantom{{}={}} \lambda x \in \dom(\theta). \sum_{y \in \var}
      \occ(y,\theta(x)) \cdot \sum_{n \in N_{\mG}} l_{\mG}(n)(y) \\
      &= \lambda y \in U \setminus \dom(\theta).  \sum_{n \in
        N_{\mG}} l_{\mG}(n)(y) \multisum \lambda x \in \dom(\theta).
      \sum_{n \in N_{\mG}} \chi(l_{\mG}(n),\theta(x)) \enspace .
    \end{split}
  \]
  Since $\mathcal G$ is a sharing graph, the total in-degree
  for $G^i$, \ie $\sum_{n \in N_{\mG}}
  \chi(l_{\mG}(n),t_i)$, is equal to the total out-degree
  $\sum_{n \in N_{\mG}} \chi(l_{\mG}(n),x_i)$. Hence
  \[
    \begin{split} 
      & \eta^{-1}(w_{\bar{n}})|_{U}\\
      =\ &\lambda y \in U\setminus \dom(\theta).  \sum_{n \in N_{\mG}}
      l_{\mG}(n)(y) \multisum
      \lambda x \in \dom(\theta).  \sum_{n \in N_{\mG}} \chi(l_{\mG}(n),x)\\
      =\ &\lambda x \in U.  \sum_{n \in N_{\mG}} l_{\mG}(n)(x)\\
      =\ &\res(\mG) = B \enspace .
    \end{split}
  \]
  This concludes the proof.
\end{proof}

The previous proof requires $\vars(\theta) \subseteq U$. However, the 
same construction also works when this condition does not hold.
\begin{example}
  Let $U=\{x,y\}, S=\{x^2,x^2y\}, \theta=\{x/s(y,z)\}$ and assume that we 
  want to compute $\mgu_\lp([S]_U,\theta)$. By extending the domain of 
  variables of interest to $V=\{x,y,z\}$, we obtain $[S']_V= 
  [x^2,x^2y,z]_{x,y,z}$. One of the 
  sharing graphs for $\theta$ and $[S']_V$ is
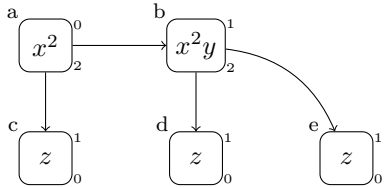
\begin{equation*}
\begin{tikzpicture}[shgraph,x=2cm,y=1.5cm]
\shnode{0}{2} a (x2) at (0,0) {$x^2$};
\shnode{1}{2} b (x2y) at (1,0) {$x^2y$};
\shnode{1}{0} c (za) at (0,-1) {$z$};
\shnode{1}{0} d (zb) at (1,-1) {$z$};
\shnode{1}{0} e (zc) at (2,-1) {$z$};
\path (x2) edge (x2y);
\path (x2) edge (za);
\path (x2y) edge (zb);
\path (x2y) edge[bend left] (zc);
\end{tikzpicture}
\end{equation*}
Following the proof of the previous theorem, we obtain the substitution
\[
\delta'=\{x/s(r(w_a),r(w_a,w_b,w_b)), y/r(w_b), z/r(w_c,w_d,w_e)\}
\enspace,
\]
where $[S']_V \rightslice [\delta']_V$ and $x^4yz^3 \in 
\alpha_\lp(\mgu_\lp([\delta']_V,\theta))$. However, what we are looking for is 
a substitution $\delta$ such that $[S]_U \rightslice [\delta]_U$ and $x^4yz^3 
\in \alpha_\lp(\mgu_\lp([\delta]_U,\theta))$. Nonetheless, we may choose 
$\delta=\delta'$ (or, if we prefer, $\delta=\delta'|_{\{x,y\}}$) to get the 
required substitution.
\end{example}

This is not a fortuitous coincidence. We will prove that it happens 
consistently, and therefore $\mgup([S]_U, \theta)$ is optimal even 
when $\vars(\theta) \nsubseteq U$.

Note that it is not an obvious result. The operation 
$\mgup([S]_U,\theta)$ is designed by first extending the set of variables of 
interest of the abstract object in order to include all the variables in 
$\vars(\theta) \setminus U$ and then performing the real operation. This 
construction does not always yield optimal operators. For example, 
\citeN{AmatoS09sharing} show that this is not the case for $\Sharing$.

\begin{theorem}[Optimality of $\mgup$]
  The abstract parallel unification $\mgup([S]_U,\theta)$ is optimal, that is:
    \[
  \forall B\in\mgup ([S]_U,\theta)~\exists [\delta]_U \in \Isubst_{\sim}.~[S]_U \rightslice 
  [\delta]_U \text{ and } B \in \alpha_\lp(\mgu([\delta]_U,\theta)) \enspace .
   \]
\end{theorem}
\begin{proof}
  Given $[S]_U \in \Linp$ and $\theta \in \Isubst$, proving optimality amounts 
  to show that, for each $B \in \mgup([S]_U,\theta)$, there is $[\delta]_U\in \Isubst_{\sim}$ 
  such that $[S]_U \rightslice [\delta]_U$ and $B \in 
  \alpha_\lp(\mgu([\delta]_U,\theta))$. By definition  $B \in 
  \mgup([S]_U,\theta)$  iff $B\in \mgup(S',\theta)$ for $S'=S \cup \{ 
  \multil v \multir \mid v \in \vars(\theta)\setminus U \}$.
  In the rest of the proof, assume $\theta =\{x_1/t_1, \ldots, x_p/t_p\}$,
  $V=U \cup \vars(\theta)$ and $B\in \mgup(S',\theta)$.
  
  Using the previous theorem, we find $\delta$ such
  that $B \in \alpha_\lp(\mgu([\delta]_V,\theta))$ and $[S']_V
  \rightslice [\delta]_V$. We want to prove that $[S]_U \rightslice
  [\delta]_U$ and $B \in \alpha_\lp(\mgu([\delta]_U,\theta))$. 
  
  We first prove $[S]_U \rightslice [\delta]_U$. Given any $v \in \var$, since
  $[S']_V \rightslice [\delta]_V$, we have $\delta^{-1}(v) \cap V \in S'$.
  There are two cases: either $\delta^{-1}(v) \cap V \in S$ or $\delta^{-1}(v) 
  \cap V = \multil w \multir$ for some $w \in V$. In the first case, 
  $\supp{\delta^{-1}(v) \cap V}\subseteq U$, hence $\delta^{-1}(v) \cap U \in 
  S$. In the latter, $\delta^{-1}(v) \cap U  = \emptymulti \in S$. Therefore
  $[S]_U \rightslice [\delta]_U$.
  
%  Now we prove that $B \in \alpha_\lp(\mgu([\delta]_V,\theta))$. We consider 
%  two cases. First of all, assume there is a variable $v \notin U$ such that 
%  $v \in \rng(\theta)$ but $v \notin \rng(\theta|_{U})$. Consider a possible 
%  sharing graph $\mG$ for $S'$ and $\theta$ and suppose a node labeled by $v$ 
%  appear in the sharing graph. Let $x_i/t_i$ be a binding where $v \in t_i$.
%  Since $x_i \notin U$, the only nodes in the layer $i$ with a positive 
%  out-degree are those labeled by $x_i$. Each node labeled $x_i$ is connected 
%  to a different node, and each node labeled $v$ is connected to a node 
%  labeled $x_i$. Since $x_i$ does not appear with positive in- / out- degree 
%  in any other layer, and $\flat{\mG}$ should be connected, the only 
%  possibility is that 
%  

  In order to prove that $B \in \alpha_\lp(\mgu([\delta]_U, \theta))$ we need 
  to  study the relationship between $\mgu([\delta]_U, \theta) = 
  [\mgu(\delta|_U, \theta)]_V$ and $\mgu([\delta]_V, \theta) = [\mgu(\delta, 
  \theta)]_V$. We split $\delta$ into $\delta|_{U \cup \rng(\theta)}$ and 
  $\delta|_{\dom(\theta) \setminus U}$. With the same considerations which 
  led to \eqref{eq:opt-proof1}, we have:
  \[
    \begin{split}
      \mgu\bigl(\delta,\theta) &= \mgu(\eq(\theta) \cup \eq(\delta|_{U \cup 
      \rng(\theta)}) \cup
      \eq(\delta|_{\dom(\theta) \setminus U})\bigr)\\
      &=\mgu\bigl(\eq(\theta) \cup \eq(\delta|_{U \cup 
            \rng(\theta)}) \cup \{ x_i = f^i(\delta(t_i)) \mid x_i \in 
            \dom(\theta) \setminus U \})\\
      &=\mgu\bigl(\eq(\theta) \cup 
            \eq(\delta|_{U \cup  \rng(\theta)}) \cup \{ \delta(t_i) = 
            f^i(\delta(t_i)) \mid  x_i  \in \dom(\theta) \setminus U \}\bigr)
            \enspace .
    \end{split}
  \]
  If $x_i \notin U$, then $x_i$ appears in $S'$ only in the multiset $\multil 
  x_i \multir$. Hence $\delta(x_i)=f^i(\delta(t_i))$ is linear and independent 
  from the  other variables, i.e., no variables in $f^i(\delta(t_i))$ appear 
  in  either $\theta$ or other bindings in  $\delta$.  As a result, 
  $\mgu(\delta,\theta)$ may be rewritten as
  \begin{equation*}
    \beta \circ \mgu \bigl(\eq(\theta) \cup \eq(\delta|_{U \cup \rng(\theta)})
    \bigr) \enspace ,
  \end{equation*}
  where $\beta$ is a substitution such that
  $\dom(\beta)=\rng(\delta|_{\dom(\theta) \setminus U})$.
%   Note that, since 
%  $\dom(\beta) \cap V=\emptyset$, then 
%  \[
%  \alpha_\lp([\beta \circ \mgu (\eq(\theta) \cup \eq(\delta|_{U \cup 
%  \rng(\theta)}) )]_V) = \alpha_\lp([\mgu (\eq(\theta) \cup 
%  \eq(\delta|_{U \cup \rng(\theta)}) )]_V)
%  \]
  
  We now split $U \cup \rng(\theta)$ into $U$, $U_1$ and $U_2$ where 
  $U_1 = (\rng(\theta) \setminus U) \cap \vars(\theta(U))$ and $U_2= 
  (\rng(\theta) \setminus U) \setminus \vars(\theta(U))$. If $y \in U_1$ there
  exists $i_y \in [1,p]$ and a position $\xi_y$ such that $x_{i_y} \in U$,
  $\theta(x_{i_y})(\xi_i)= t_{i_y}(\xi_y)=y$ and 
  $\delta(x_{i_y})(\xi_y)=\delta^{i_y}_{\xi_y}(y)$. 
  Since $y \notin U$,   then $\delta(y)$ is linear  and
  independent from $\theta$ and  the other bindings in $\delta$. Therefore
    \[
      \begin{split}
        &\mgu\bigl(\eq(\theta) \cup \eq(\delta|_{U \cup U_2}) \cup 
        \eq(\delta|_{U_1})\\
        &= \mgu\bigl(\eq(\theta) \cup \eq(\delta|_{U \cup U_2}) \cup \{ 
        \delta^{i_y}_{\xi_y}(y)=\delta(y)
        \mid y \in U_1 \}\bigr)\\
        &\qquad \text{[by equations $x_{i_y}=t_{i_y}$ in $\eq(\theta)$ and
         $x_{i_y} = f^{i_y}(\delta(t_{i_y}))$ in $\eq(\delta|_U)$,}\\
        &\qquad \text{\phantom{]}restricted 
         to position $\xi_y$]}\\
        &= \beta' \circ \mgu\bigl(\eq(\theta) \cup \eq(\delta|_{U \cup U_2}) 
        \bigr)\\
        &\qquad \text{[linearity and independence of $\delta(y)$]}\\        
      \end{split}
    \]
    where $\beta'=\mgu(\{ \delta^{i_y}_{\xi_y}(y)=\delta(y)\mid y \in U_1 
    \})$ and  $\dom(\beta')=\rng(\delta|_{U_1})$. 
%    As before, we have that
%    \[
%      \alpha_\lp([\beta' \circ \mgu (\eq(\theta) \cup \eq(\delta|_{U \cup 
%      U_2} ))]_V) = \alpha_\lp([\mgu (\eq(\theta) \cup 
%      \eq(\delta|_{U \cup U_2}) )]_V)
%    \]
%    
   We may proceed as follows:
   \[
     \begin{split}
      &\mgu(\eq(\theta) \cup \eq(\delta|_U) \cup 
      \eq(\delta|_{U_2}))\\
      &= \mgu(\eq(\theta|_U) \cup \eq(\theta|_{\dom(\theta) \setminus U}) \cup
      \eq(\delta|_U) \cup \eq(\delta|_{U_2}))\\
      &= \theta|_{\dom(\theta) \setminus U} \circ \mgu(\eq(\theta|_U) \cup 
         \eq(\delta|_U) \cup \eq(\delta|_{U_2}))\\
         &\qquad \text{[since $\dom(\theta) \setminus U$ is disjoint from the 
         variables in $\theta|_U$, $\delta|_U$, $\delta_{U_2}$]}\\
      &=\theta|_{\dom(\theta) \setminus U} \circ \mgu(\theta|_U, \delta|_U) 
      \circ \delta|_{U_2}\\
      & \qquad \text{[since $\vars(\delta|_{U_2})$ is disjoint from 
        $\vars(\theta|_U) \cup \vars(\delta|_U)$]} \enspace .
    \end{split}
  \]
  Note that $\theta|_{\dom(\theta) \setminus U} \circ \mgu(\theta|_U, 
  \delta|_U)$ is $\mgu(\delta|_U,\theta)$, and will be denoted in the 
  following by $\eta$. Therefore, we have
  \[
  \mgu(\delta, \theta) = \beta \circ \beta' \circ \eta \circ \delta|_{U_2}
  \enspace .
  \]
%          
%  Note that, since $\vars(\delta|_U) \cap \vars(\theta) \subseteq U$, then
%  $\mgu([\delta]_U, \theta) = [\eta]_V$. We prove that $B 
%  \in \alpha_\lp([\eta]_V)$. 
  
  We check that $B \in \alpha_\lp([\eta]_V)$. Let $v$ be the variable 
  such that $(\beta \circ \beta' \circ \eta \circ \delta|_{U_2})^{-1}(v) 
  \cap V= B$. We want to find $\bar v$ such that $\eta^{-1}(\bar v) \cap 
  V=B$. First of all, since $\beta$ and $\beta'$ have no variables in 
  common with $V$, then $(\beta \circ \beta' \circ \eta \circ 
  \delta|_{U_2})^{-1}(v) \cap V =  (\eta \circ \delta|_{U_2})^{-1}(v) \cap 
  V$. If $v \notin \vars(\delta|_{U_2})$, then $(\eta \circ 
  \delta|_{U_2})^{-1} = \eta^{-1}( v )$ and we get the 
  required result  with $\bar v=v$. If $v \in \rng(\delta|_{U_2})$ we know
  that $v$ only occurs once in $\delta|_{U_2}$ and never in 
  $\eta$.  Then $(\eta \circ \delta|_{U_2})^{-1}(v)=
  \eta^{-1}(y) \multisum \eta^{-1}(v) = \eta^{-1}(y) \multisum \multil v
  \multir$ for the unique $y \in U_2$ such that $v \in 
  \vars(\delta(y))$.  Therefore, since $v \notin V$, we may choose $\bar v=y$.
  Finally, if $v \in \dom(\delta|_{U_2})$ then $(\eta
  \circ \delta|_{U_2})^{-1}(v)=\emptymulti$ and we may take $\bar v$ to be
  any variable not in $V \cup \vars(\eta)$.
\end{proof}

\section{Related Work}

Proving optimality results for abstract unification operators on domains 
involving sharing information is a difficult task.
The well-known domain $\Sharing$ (without any linearity or freeness 
information) is the unique domain whose abstract unification operator has been 
proved optimal in the general case of multi-binding substitutions 
(\citeN{CF99}, later extended by \citeN{AmatoS09sharing} to substitutions with 
variables out of the set of interest).

In the simpler case of single-binding unification, the only optimality results 
for domains combining sharing and linearity appeared in 
\cite{AmatoS09sharlin}. 

As far as we known, this is the first optimality result for domains involving 
linearity information for multi-binding 
substitutions.

Although $\Linp$ is not amenable to a direct 
implementation, as future work  we plan to design 
suitable abstractions
using numerical domains. The idea is to consider $\omega$-sharing groups  with symbolic multiplicities constrained by linear inequalities, such as $x^\alpha y^\beta$ with $\alpha=\beta+2$.
We plan to implement in our analyzers
Random \cite{AmatoS12lpar,AmatoPS10-rv} and Jandom \cite{AmatoS-SOAP13} an abstract domain based on (template) parallelotopes 
(\citeNP{AmatoPS11JSC}; \citeyearNP{AmatoS12-entcs,AmatoPS10-sas,AmatoLM09}), exploiting the recent localized \cite{AmatoS13sas} 
iteration strategies.

%\appendix 

%\bibliographystyle{acmtrans} 
%\bibliography{asbiblio}

\end{document}